\documentclass{ieeeaccess}
\usepackage{cite}
\usepackage{amsmath,amssymb,amsfonts}
\usepackage{graphicx}
\usepackage{textcomp}
\usepackage{soul, xspace}
\usepackage[normalem]{ulem}
\usepackage{caption}
\usepackage{xcolor}
\usepackage{xspace}
\usepackage[normalem]{ulem}
\usepackage{multirow}
\usepackage{algorithmicx,algpseudocode}
\usepackage{soul}
\usepackage{subfigure}
\usepackage{comment}
\usepackage{enumitem}
\setlist[itemize]{noitemsep, topsep=0pt}
\usepackage{amsthm}
\usepackage[hyphens]{url}
\usepackage{mathtools}
\usepackage{circledsteps}
\usepackage{algorithm,amsmath,lipsum,xcolor,caption}

\usepackage{hyperref}
\hypersetup{
    colorlinks=true,
    linkcolor=black,
    citecolor=black,
    urlcolor=black,
    filecolor=black,
    pdftitle={\@title},
}

\def\BibTeX{{\rm B\kern-.05em{\sc i\kern-.025em b}\kern-.08em
    T\kern-.1667em\lower.7ex\hbox{E}\kern-.125emX}}
\begin{document}
\history{Received 4 July 2024, accepted 25 July 2024, date of publication 5 August 2024, date of current version 16 August 2024.}
\doi{10.1109/ACCESS.2024.3438996}

\definecolor{bindon-red}{RGB}{112, 0, 0}
\newcommand{\hyoung}[1]{\textcolor{red}{Hyoung:  #1}}
\newcommand{\hyunmin}[1]{\textcolor{blue}{hyunmin:  #1}}
\newcommand{\seungho}[1]{\textcolor{magenta}{seungho: #1}}
\newcommand{\bindon}[1]{\textcolor{bindon-red}{bindon: #1}}
\newcommand{\jihun}[1]{\textcolor{orange}{Jihun: #1}}
\newcommand{\seonhye}[1]{\textcolor{teal}{seonhye: #1}}
\newcommand{\sysname}{UniHENN\xspace}
\newcommand{\tenseal}{TenSEAL\xspace}
\newcommand{\pycrcnn}{PyCrCNN\xspace}
\newcommand{\lenet}{LeNet\xspace}
\newcommand{\relu}{ReLU\xspace}
\newcommand{\databoxframe}{\texttt{Data Box Frame}\xspace}

\title{\sysname: Designing Faster and More Versatile Homomorphic Encryption-based CNNs without \texttt{im2col}}


\author{
\uppercase{Hyunmin Choi}\authorrefmark{1} \authorrefmark{2}, 
\uppercase{Jihun Kim\authorrefmark{3},
\uppercase{Seungho Kim}\authorrefmark{1},
\uppercase{Seonhye Park}\authorrefmark{1},
\uppercase{Jeongyong Park}\authorrefmark{1},
\uppercase{Wonbin Choi}\authorrefmark{2},
\uppercase{Hyoungshick Kim}\authorrefmark{1}}
\address[1]{Department of Electrical and Computer Engineering, Sungkyunkwan University, Seoul 03063, Republic of Korea}
\address[2]{NAVER Cloud Security Dev., Seongnam-si 463-824, Republic of Korea}
\address[3]{Department of Mathematics, Sungkyunkwan University, Seoul 03063, Republic of Korea}
}
\tfootnote{This work was supported in part by NAVER Cloud Corporation; and in part by Korean Government through Institute of Information \&
communications Technology Planning \& Evaluation (IITP) under Grant RS-2022-II220688, Grant 2022-0-01199, Grant RS-2023-00229400, Grant RS-2024-00419073, and Grant RS-2024-00436936.}

\markboth
{H. Choi \headeretal: \sysname: Designing Faster and More Versatile Homomorphic Encryption-based CNNs without \texttt{im2col}}
{H. Choi \headeretal: \sysname: Designing Faster and More Versatile Homomorphic Encryption-based CNNs without \texttt{im2col}}

\corresp{Corresponding author: Hyoungshick Kim (hyoung@skku.edu).}

\begin{abstract}
Homomorphic encryption (HE) enables privacy-preserving deep learning by allowing computations on encrypted data without decryption. However, deploying convolutional neural networks (CNNs) with HE is challenging due to the need to convert input data into a two-dimensional matrix for convolution using the \texttt{im2col} technique, which rearranges the input for efficient computation. This restricts the types of CNN models that can be used since the encrypted data structure must be compatible with the specific model. \sysname is a novel HE-based CNN architecture that eliminates the need for \texttt{im2col}, enhancing its versatility and compatibility with a broader range of CNN models. \sysname flattens input data to one dimension without using \texttt{im2col}. The kernel performs convolutions by traversing the image, using incremental rotations and structured multiplication on the flattened input, with results spaced by the stride interval. Experimental results show that \sysname significantly outperforms the state-of-the-art 2D CNN inference architecture named PyCrCNN in terms of inference time. For example, on the LeNet-1 model, \sysname achieves an average inference time of 30.089 seconds, about 26.6 times faster than PyCrCNN's 800.591 seconds. Furthermore, \sysname outperforms TenSEAL, an \texttt{im2col}-optimized CNN model, in concurrent image processing. For ten samples, \sysname (16.247 seconds) was about 3.9 times faster than TenSEAL (63.706 seconds), owing to its support for batch processing of up to 10 samples. We demonstrate \sysname's adaptability to various CNN architectures, including a 1D CNN and six 2D CNNs, highlighting its flexibility and efficiency for privacy-preserving cloud-based CNN services.

\end{abstract}


\titlepgskip=-15pt

\maketitle

\section{Introduction}
\label{sec:introduction}

The widespread adoption of deep learning has expanded beyond traditional tasks like image classification~\cite{xu2020social} to diverse applications such as attack detection~\cite{tian2019distributed} and resource extraction analysis~\cite{xing2023coal}, with machine learning (ML) now playing a crucial role in sensitive sectors including finance, healthcare, and autonomous systems. However, this proliferation has raised significant privacy concerns~\cite{liu2021machine,zhang2020privacy}, particularly for cloud-based ML services processing sensitive data on remote servers~\cite{abuadbba2020use}. Balancing ML's powerful capabilities with the need to safeguard individual privacy has become a critical challenge, highlighting the urgent need for advanced privacy-preserving techniques in ML applications handling sensitive information.

Homomorphic encryption (HE) is a powerful tool for preserving privacy in ML applications~\cite{kim2018logistic,HE_Accurate_CNN,Lee,Dathathri,kim2020efficient,kim2018poster}. For example, Choi et al. introduced privacy-preserving biometric authentication systems \cite{choi2024blind,blindmatch} that leverage HE for web and cloud environments. HE allows computations to be performed on encrypted data without decryption, ensuring that sensitive information remains confidential even when processed by a third party. This is particularly beneficial for cloud-based ML services, where sensitive data is often transmitted to remote servers for processing. By using HE, cloud providers can ensure data privacy, such as a bank training an ML model to detect fraudulent transactions without accessing customer data or a healthcare provider analyzing patient medical records without disclosing patient identities. Additionally, faster inference time with privacy preservation is crucial for real-time applications in these sensitive fields, where quick and secure data processing can significantly improve decision-making and user experience.


Convolution operations are essential to Convolutional Neural Networks (CNNs) but require significant computational resources when performed on ciphertexts in HE. Each kernel shift requires element-wise multiplications and additions, resource-intensive operations on ciphertexts. To reduce the computational load, most HE-based CNN implementations like TenSEAL~\cite{tenseal2021} use the \texttt{im2col} function. This function transforms the input data into a matrix without discarding or altering the original information, thus allowing for more efficient computation. TenSEAL is applied in various research areas, including healthcare~\cite{khalid2023privacy}, federated learning~\cite{wang2022federated}, and biometrics~\cite{yang2023review}. However, this approach limits the versatility of CNN models to configurations that accommodate only a single convolution layer, potentially hindering their use in more complex structures. In practical scenarios, encrypted user data could be useful across multiple ML models. For instance, an encrypted patient image stored in a hospital database could be utilized for future analysis by various ML algorithms. Therefore, encrypting data without restricting it to a specific model would offer greater flexibility and utility.

We introduce \sysname, a privacy-preserving CNN model using HE that enables efficient CNN inference without relying on the \texttt{im2col} function. Our novel convolution algorithm flattens input data into a one-dimensional form, eliminating the need for data rearrangement. The kernel traverses the input, performing convolutions with incremental rotations and structured multiplication, ensuring only relevant elements are used. This approach overcomes \texttt{im2col} limitations by reducing image size dependency and focusing on kernel size, improving efficiency and flexibility across various ML models. \sysname also supports batch operations for efficient multi-ciphertext processing. Our key contributions include:

\begin{itemize}[leftmargin=*]
    \item We introduce \sysname, a novel CNN model inference mechanism based on HE, which facilitates input ciphertext reusability. Unlike other CNN implementations that require a specific model input structure for the \texttt{im2col} function, \sysname is designed to handle model-free input ciphertexts, enabling the use of encrypted input data across various HE-based ML services without the need for re-encryption. We demonstrate the effectiveness of this approach by successfully constructing and training seven different CNN models on four datasets: MNIST \cite{mnist}, CIFAR-10 \cite{cifar10}, USPS \cite{uspsdataset}, and electrocardiogram (ECG) \cite{ecgdataset}. The source code for \sysname is available at~\url{https://github.com/hm-choi/uni-henn}.
    
    \item We empirically demonstrate the efficiency of \sysname. Experimental results indicate an average inference time of 30.089 seconds, significantly outperforming the state-of-the-art HE-based 2D CNN inference architecture PyCrCNN~\cite{disabato2020privacy}, which requires an average of 800.591 seconds for inference.
    
    \item We introduce a batch processing strategy for \sysname to handle multiple data instances in a single operation. This strategy efficiently combines multiple data instances into a single ciphertext, reducing the inference time for 10 MNIST images to 16.247 seconds. It outperforms TenSEAL's CNN model~\cite{tenseal2021}, which takes 63.706 seconds. While \sysname is less efficient than TenSEAL for processing a single image, it surpasses TenSEAL's performance when processing multiple images simultaneously, particularly when $k \geq 3$.
\end{itemize}
\section{Background}
\label{sec:backgrounds}


\subsection{Convolutional Neural Network (CNN)}

Convolutional Neural Networks (CNNs) were first introduced by LeCun et al.~\cite{lecun1989backpropagation} for processing grid-structured data like images. In 1989, the LeNet-1 model was presented~\cite{lecun1989handwritten}, comprising two convolutional layers, two average pooling layers, and one fully connected layer. Which is the first concept of the LeNet architecture.

In 1998, the LeNet-5 model was presented~\cite{lecun1998gradient}, comprising three convolutional layers, two pooling layers, and two fully connected layers. It was designed to efficiently recognize handwritten postal codes, outperforming traditional methods like the multilayer-perceptron model.

CNNs gained significant traction in the 2010s, primarily due to their excellent performance on large datasets like ImageNet~\cite{imagenet}, which includes millions of images across 1,000 classes. 

A typical CNN consists of an input layer, multiple hidden layers, and an output layer. The input layer receives the image and forwards its pixel values into the network. Since images are generally matrices, the node count in the input layer corresponds to the image size. The hidden layers, situated between the input and output layers, extract relevant features using convolution, activation, and pooling operations. Each hidden layer receives information from either the input or preceding layers, facilitating iterative learning. This information is then passed to subsequent layers, culminating in the final prediction or classification at the output layer.

Kiranyaz et al. \cite{7318926} introduced a 1D CNN for disease-specific ECG classification. A 1D CNN is a variant of CNNs tailored for one-dimensional data, making it particularly suitable for signal data, time-series data, and text data. In this architecture, both the input and the convolution filter are one-dimensional; the filter slides over the input in the convolution layer to execute operations.

\subsection{Homomorphic Encryption (HE)}

Homomorphic encryption (HE) is a cryptographic technique that allows computations on encrypted data without the need for decryption. Formally, given messages $m_1$ and $m_2$, an encryption function $Enc$, and computationally feasible functions $f$ and $f'$ for ciphertext and plaintext, respectively, HE satisfies $f(Enc(m_1), Enc(m_2)) = Enc(f'(m_1, m_2))$.

Several HE schemes, such as BGV~\cite{cryptoeprint:2011/277}, GSW-like schemes~\cite{brakerski2014lattice,chillotti2016faster}, and CKKS~\cite{cheon2017homomorphic}, exist, each with different computational needs and data types. In \sysname, we use the CKKS scheme, which is advantageous for encrypting vectors of real or complex numbers.

CKKS encryption requires structuring plaintext to mirror input data. The \texttt{number of slots}, defined during parameter selection, determines the encryptable vector size. CKKS supports three fundamental ciphertext operations: addition, multiplication, and rotation.

Let $N$ be the total degree of the CKKS parameter, then the \texttt{number of slots} is $N/2$. Two real vectors $\mathbf{v_1}, \mathbf{v_2}$ and $C(\mathbf{v_1}), C(\mathbf{v_2})$ denote the ciphertext of the vectors $\mathbf{v_1}, \mathbf{v_2}$. (P) indicates that the operation is defined between a ciphertext and a plaintext, and (C) indicates that the operation is defined between two ciphertexts. Then the operations, addition ($Add$), multiplication ($Mul$), and rotation $(Rot)$ can be represented as follows:

\begin{itemize}
    \item Addition (P): $Add(C(\mathbf{v_1}), \mathbf{v_2}) = C(\mathbf{v_1} + \mathbf{v_2})$
    \item Addition (C): $Add(C(\mathbf{v_1}), C(\mathbf{v_2})) = C(\mathbf{v_1} + \mathbf{v_2})$
    \item Multiplication (P): $Mul(C(\mathbf{v_1}), \mathbf{v_2}) = C(\mathbf{v_1} \times \mathbf{v_2})$
    \item Multiplication (C): $Mul(C(\mathbf{v_1}), C(\mathbf{v_2})) = C(\mathbf{v_1} \times \mathbf{v_2})$
    \item Rotation: $Rot(C(\mathbf{v}), r) = C(v_{r}, v_{r+1}, \ldots, v_{N/2 -1}, v_{0},$ $\ldots, v_{r-1})$, where $\mathbf{v} = (v_{0}, v_{1}, \ldots, v_{N/2-1})$ and $r$ is a positive integer.
\end{itemize}

The operations $+$ and $\times$ represent elementwise addition and multiplication in the plaintext space, respectively. The \texttt{depth} parameter specifies the maximum number of sequential multiplications that can be performed while ensuring correct decryption, as each multiplication operation increases ciphertext noise. Exceeding this limit may result in decryption failure. While \texttt{depth} is predetermined during parameter selection, \texttt{level} dynamically indicates the remaining multiplication capacity, decreasing with each multiplication operation executed. 

For practical HE implementations, several general-purpose HE libraries are available, including SEAL-Python~\cite{huelseseal}, Lattigo~\cite{lattigo}, HElib~\cite{halevi2014algorithms}, and OpenFHE~\cite{openFHE}. We selected SEAL-Python for its efficient support of the CKKS scheme. Building on this foundation, we developed \sysname, a HE-based framework specifically optimized for CNN inference. While currently implemented with SEAL-Python, \sysname is designed to be adaptable to other libraries.

\section{Overview of \sysname}

\newtheorem{theorem}{Theorem}
\newtheorem{lemma}[theorem]{Lemma}
\newtheorem{proposition}[theorem]{Proposition}
\newtheorem{corollary}{Corollary}[theorem]

\label{sec:introduction_to_framework}

In this section, we introduce \sysname, a HE-based framework designed for inference on encrypted data. \sysname uses the CKKS HE scheme to encrypt input data, facilitating its integration into any CNN models without needing the \texttt{im2col} function, which demands a specific input shape. To minimize the computational overhead of HE-based inference, we employ three innovative techniques:

\begin{enumerate}
    \item Unlike previous approaches~\cite{tenseal2021, disabato2020privacy, halevi2014algorithms} that rely on the input size, \sysname calculates the total number of HE operations for the fully connected layer based on the output size. As the output size is typically smaller than the input size in the fully connected layer, this method reduces the average time per operation.
    \item \sysname enables batch operations, allowing the consolidation of multiple ciphertexts into one for more efficient processing. This considerably reduces the overall computational time, marking a significant advantage for \sysname in large-scale data processing scenarios.
    \item \sysname opts for average pooling, which eliminates the need for multiplication, thereby reducing operation time. This configuration allows for more layers within given parameter settings, giving service providers more flexibility to incorporate average pooling without concern for operation count.
\end{enumerate}

\begin{figure}[t]
\centerline{\includegraphics[trim=0cm 0cm 0cm 0cm, clip=true, width=0.99\columnwidth]{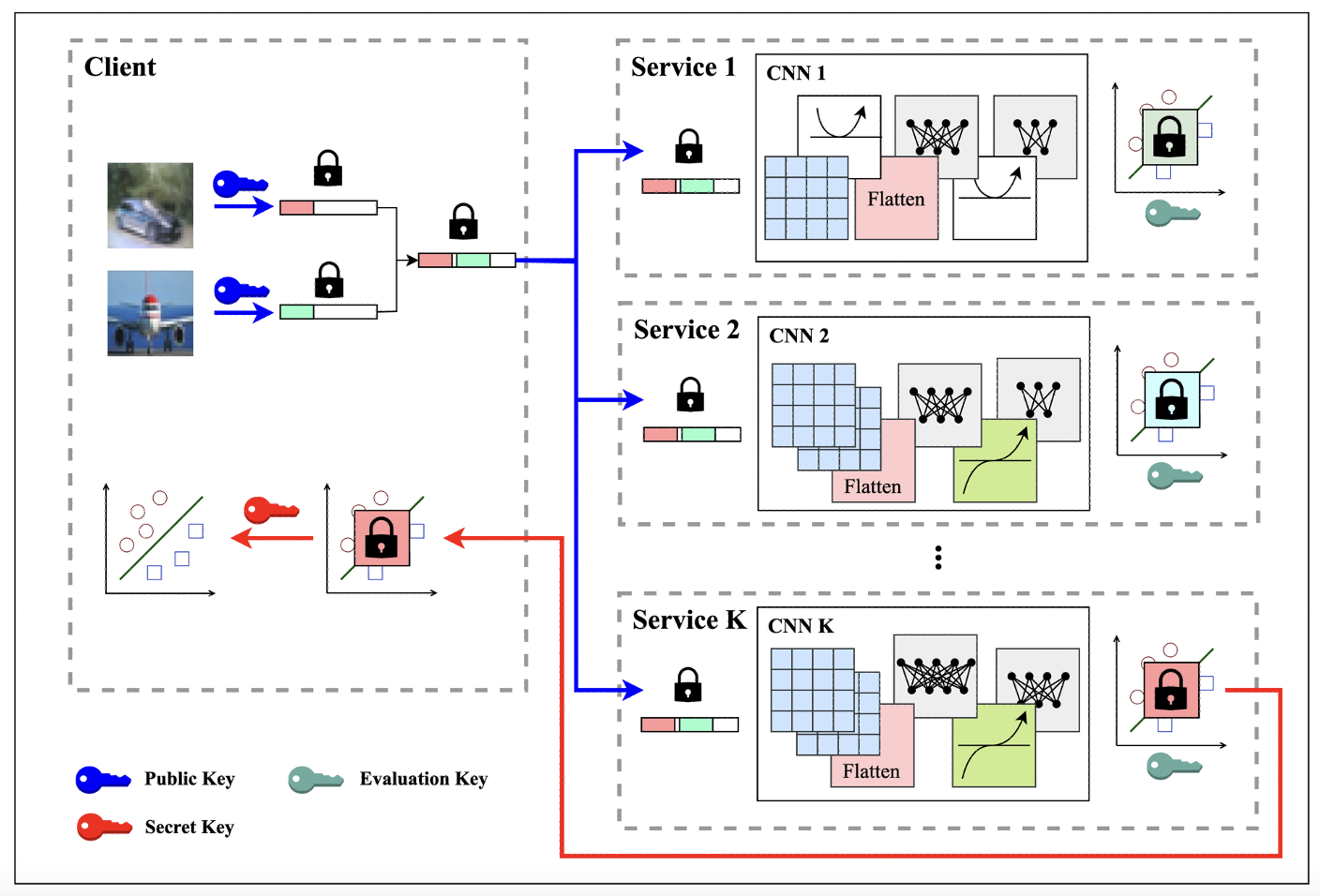}}
\caption{Overview of the \sysname architecture.}
\label{General_overview}
\end{figure}

Figure~\ref{General_overview} presents a high-level overview of \sysname's operational flow. The process begins when a client encrypts one or more images using a public key. These encrypted images are consolidated into a single ciphertext and sent to a cloud service specializing in data analytics. Each service (i.e., Service 1, Service 2, \dots, Service $K$) uses its unique CNN model for data processing, executing specific algorithms and computations on the received ciphertext with the evaluation key. Each CNN model has a distinct layer architecture and optimized parameters. After processing the ciphertext through their respective CNN models, the encrypted inference results are returned to the client. The client decrypts these results using the corresponding secret key to obtain the processed outcomes from each CNN model. This framework allows the client to benefit from multiple CNN models while preserving the confidentiality of the data, indicating a significant step forward in privacy-preserving machine learning. Table~\ref{table:notation} summarizes the notations used throughout this paper.

\begin{center}
\begin{table}[!t]
\setlength{\tabcolsep}{6pt}
\centering
\caption{Notations used for \sysname.}
\resizebox{0.99\linewidth}{!}{ 
\begin{tabular}{l p{6cm}}
\noalign{\smallskip}\noalign{\smallskip}\hline
Notation & Definition \\
\hline \hline
$N$ & Degree of polynomial. It is in the form of a power of two.\\
\hline
$D, L$ & Depth and level of ciphertext. This represents the current remaining limit on the number of multiplication operations allowed for the ciphertext. \\
\hline
$CH_{in}, CH_{out}$ & Number of input and output channels in layer \\
\hline
$DAT_{in}, DAT_{out}$ & Number of input and output data in FC layer \\
\hline
$R_{q}$ & Polynomial quotient ring $(\mathbb{Z}_q/ \langle X^N + 1 \rangle)$. It is used to create a vector space of the ciphertext. \\
\hline
$W_{img}, H_{img}$ & Width and height of the input image data \\
\hline
$K$ & Width and height of the kernel \\
\hline
$W_{in}, H_{in}$ & Width and height of input data from this layer \\
\hline
$W_{out}, H_{out}$ & Width and height of output data from this layer \\
\hline
$S_{total}$ & It is used in Section~\ref{Flatten layer converter}. The value is the multiplication of each convolutional layer's stride and the kernel sizes of all average pooling layers.  \\
\hline

\end{tabular}
}
\label{table:notation}
\end{table}
\end{center}

\renewcommand{\arraystretch}{1.3}


\subsection{Construction of Input Data}
\label{Construction of input data}

There is a variety of data types for ML services. For instance, while image data is two or three-dimensional, statistical data is usually one or two-dimensional. To handle such varying data dimensions, \sysname initially flattens the input data into a one-dimensional array in row-major format. When performing encryption, the data is located from the first of the list, and the remaining space is filled with zero to perform encryption. The remaining space will be used in the batch operations.

Figure~\ref{encryption_of_flattened_data} shows an example of data transformation for encryption. The input data consists of $W_{img} \times H_{img}$ numbers. We flatten the original data row by row, as shown in Figure~\ref{encryption_of_flattened_data}. This flattening procedure ensures universal compatibility of the input data with ML models, enhancing adaptability to algorithms used by ML services.

\begin{figure}[!ht]
\centerline{\includegraphics[trim=0cm 0.1cm 0cm 0cm, clip=true, scale=.3]{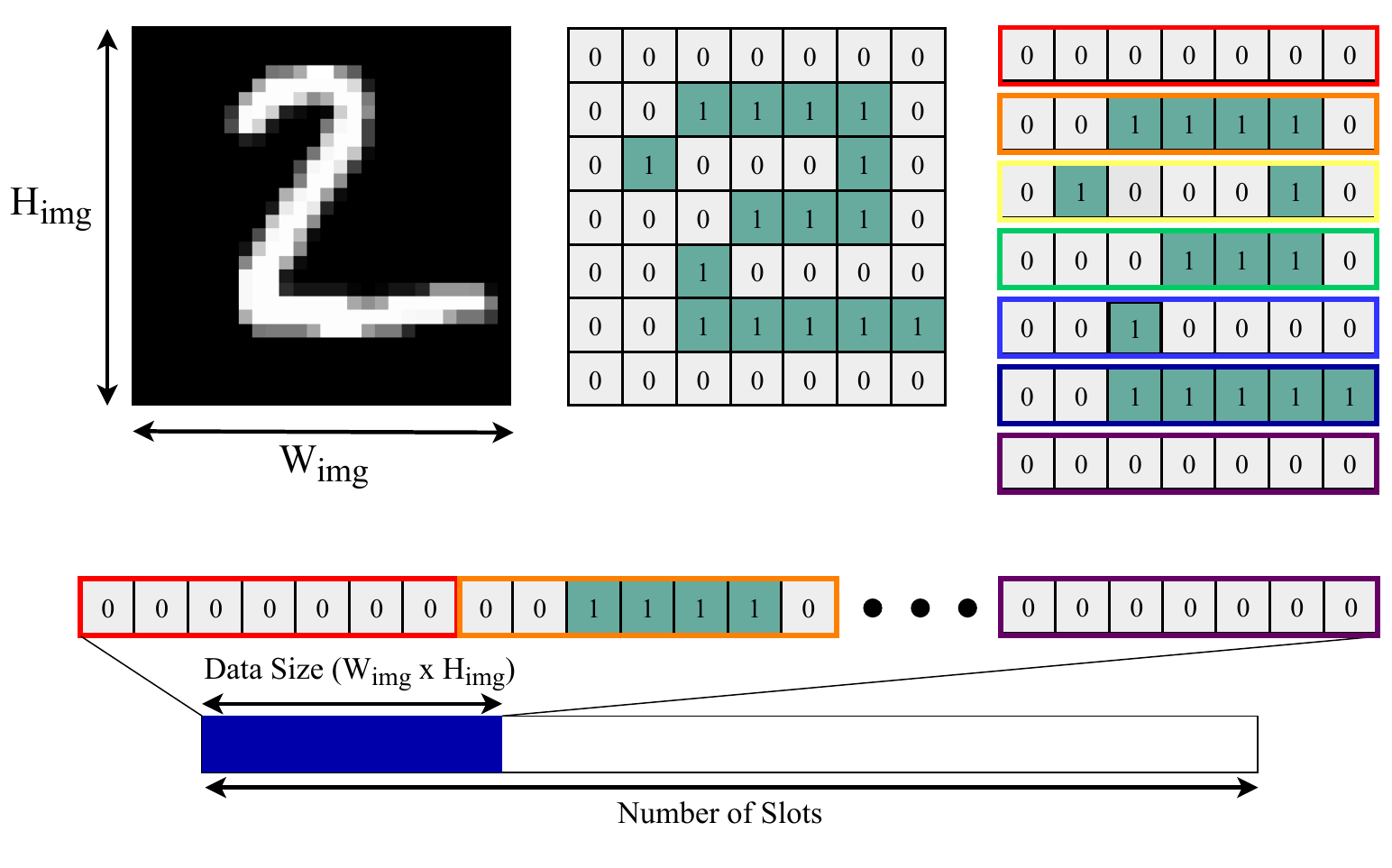}}
\caption{Data transform for \sysname.}
\label{encryption_of_flattened_data}
\end{figure}

\subsection{Combining Ciphertexts for Batch Operation}
\label{Combine input ciphertext with the compress method}

The technique in \sysname combines multiple encrypted input data into a single ciphertext for batch processing. Incorporating multiple encrypted data into a single ciphertext involves sequentially inserting each encrypted data, ensuring enough space to prevent overlap with other encrypted data, as shown in Figure~\ref{encryption_for_batch_data}. This is achieved by rotating each encrypted data before adding it to the ciphertext.

Importantly, the size allocation is not solely based on the input size of the encrypted data but also considers the size of the intermediate or final output vector generated by performing CNN operations. By examining the given CNN model structure, we can pre-calculate this size in advance, thereby preventing the overlap of input data results.

Figure~\ref{encryption_for_batch_data} illustrates one example. In this case, 800 slots are needed to process the input data through the CNN model. Thus, the $i$-th input data would be rotated to the right by $(i-1) \times 800$ slots, and the rotated vector is added to the ciphertext.
\sysname can integrate any number of ciphertexts into a single one, provided the total size of the encrypted data does not exceed the number of slots. This method enhances the efficiency of \sysname when handling multiple ciphertexts concurrently.

\begin{figure}[!ht]
\centerline{\includegraphics[trim=0cm 0.1cm 0cm 0cm, clip=true, , width=.8\columnwidth]
{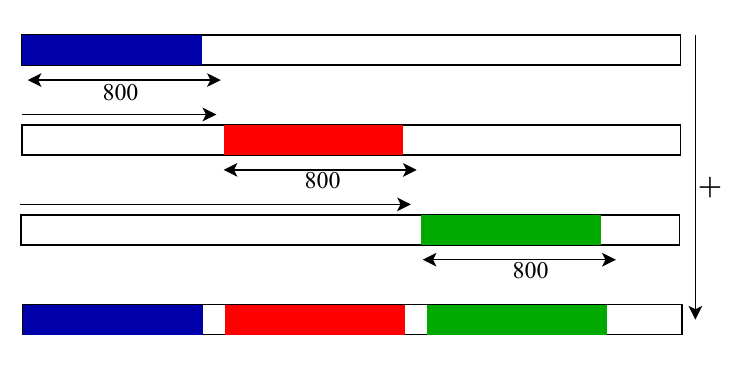}}
\caption{Combining multiple ciphertexts into one.}
\label{encryption_for_batch_data}
\end{figure}

\section{CNN Model Construction}

\subsection{Construction of the Drop-DEPTH}
The parameter is pre-defined because UniHENN's design philosophy is to support the ciphertext that is encrypted before the model is fixed. So, the CKKS parameter is determined and encrypted before the CNN model is defined for reusability of the input ciphertext independent of the model. But, the multiplication and rotation time of HE is increased dependent on the \textit{level}. Thus, the best choice of the CKKS parameter is determined by the number and structure of the CNN model's layers. However, in this scenario, the \textit{depth} is defined before when the CNN model is fixed, which means that the \textit{depth} can be larger than the exact number of total levels of the CNN's model. Thus, we suggest the \textit{Drop-Level}, a novel approach to reduce the \textit{level} that is optimized when the target CNN model. If the input ciphertext has \textit{depth} $D$ and the target CNN model's level is $L$ where $L \leq D$ then the \textit{Drop-Level} reduces the \textit{level} of the input ciphertext as $L$ with multiply $D-L$ ciphertexts that are the encryption of vectors where all elements are consists of $1$. By these simple structures, the time of execution is very fast (In our experiments in Section~\ref{experiment}, the time of \textit{Drop-Level} operation does not exceed $100 ms$. 

\subsection{Construction of the Convolutional Layer}
\label{Construction of the Convolutional Layer}

The \texttt{im2col} encoding is a popular and efficient algorithm for transforming multidimensional data into matrix form to facilitate convolution operations. However, it requires precise arrangement of input matrix elements, which is challenging to achieve on ciphertexts. Therefore, TenSEAL~\cite{tenseal2021}, a sophisticated open-source HE-enabled ML SDK, performs this operation in a preprocessing step before encryption, allowing only a single convolutional layer.

To simplify this process and improve flexibility, we propose a new method for constructing convolutional layers without relying on \texttt{im2col}. Our approach uses flattened input data, eliminating the need for a specific arrangement of input matrix elements as follows:

\begin{enumerate}
    \item \textbf{Kernel and Stride}: In convolutional layers, a small matrix called a kernel moves over the input image. Each movement of the kernel is called a stride. At each position, the kernel multiplies its values with the corresponding values in the input image and sums the results.
    
    \item \textbf{Input and Kernel Dimensions}: The dimensions of the input image are denoted as $(W_{img}, H_{img})$, with a stride of $S$. For simplicity, we assume the kernel has equal width and height, both represented by $K$.
    
    
    \item \textbf{Flattening the Input}: In the context of HE, we flatten the two-dimensional input data into a one-dimensional form. This simplifies the convolution process under encryption.
    
    \item \textbf{Rotation and Multiplication}: The convolution operation involves rotating the input data and multiplying it with the kernel. The number of rotations equals the number of elements in the kernel. Each element in the input data is multiplied by the corresponding element in the kernel, and the results are summed up.
    
    \item \textbf{Output}: The output of the convolution is an array with values spaced according to the stride interval.
\end{enumerate}

Figures \ref{figure:BasicConv} and \ref{figure:Conv2He} illustrate the convolution operations on two-dimensional and flattened input data, respectively.

\begin{figure}[!ht]
\centerline{\includegraphics[trim=0cm 0.1cm 0cm 0cm, clip=true, width=1\columnwidth]{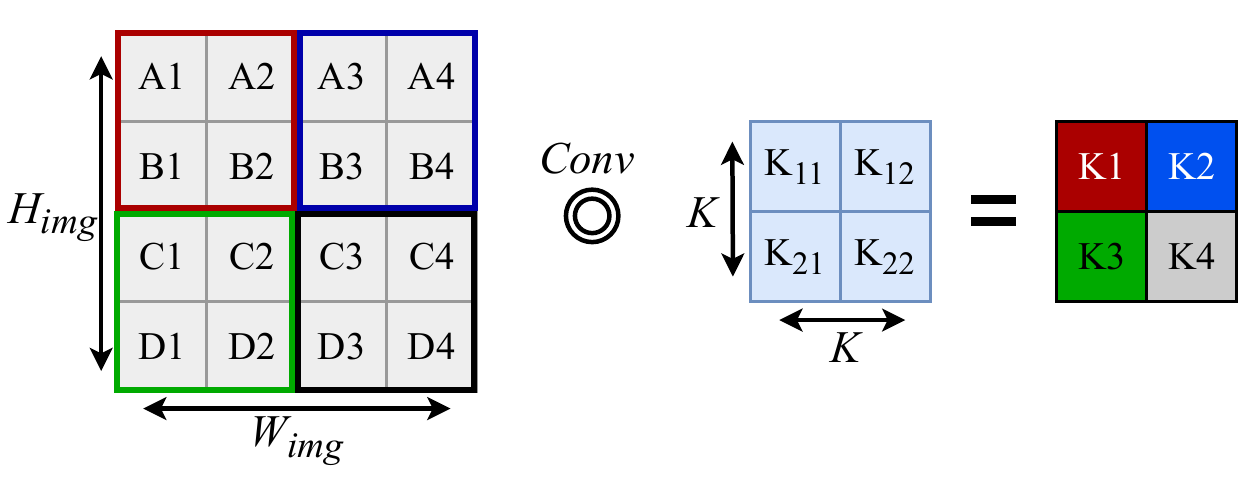}}
\caption{Convolution operations on the two-dimensional input data and kernel.}
\label{figure:BasicConv}
\end{figure}

\begin{figure}[!ht]
\centerline{\includegraphics[width=1\columnwidth]{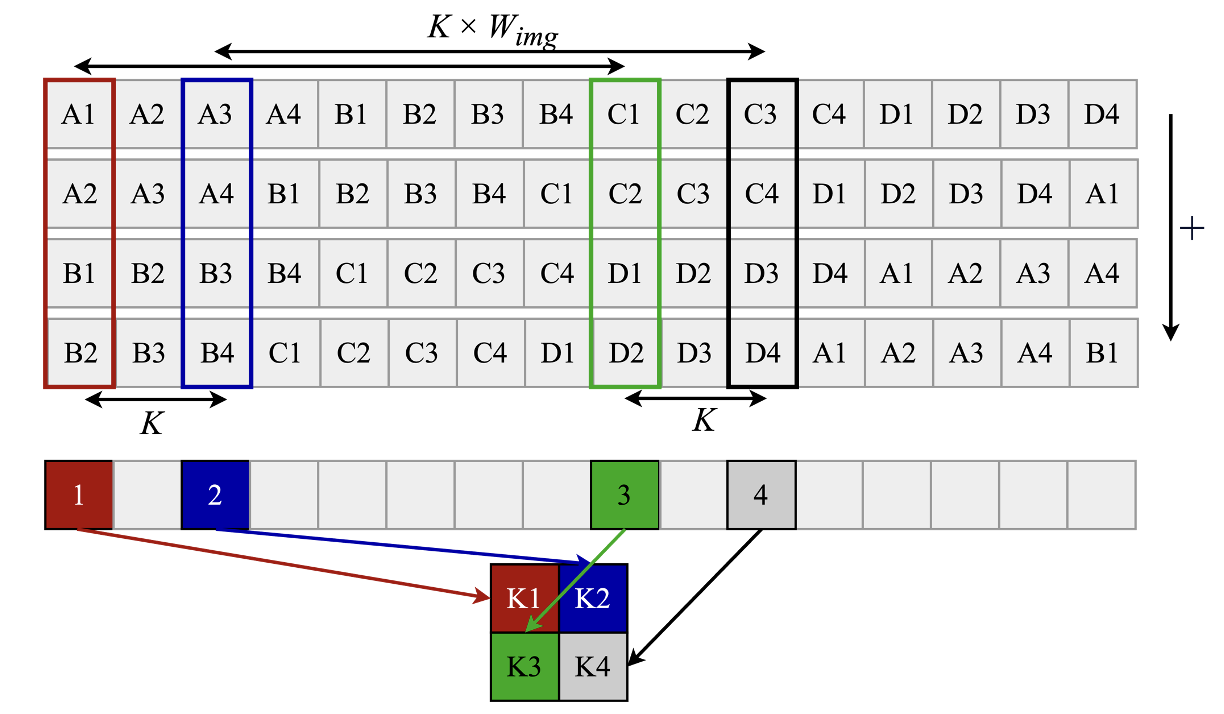}}
\caption{Convolution operations on the flattened data.}
\label{figure:Conv2He}
\end{figure}

Figure~\ref{figure:Conv2He} illustrates the convolution process: the input data is incrementally shifted by $S$, $K$ times, with leftward rotations after each shift. These rotated ciphertexts are then multiplied by a specially structured kernel vector containing a single kernel element, effectively eliminating unrelated elements during convolution. Algorithm~\ref{algorithm:convolution layer} provides a detailed description of this convolution layer construction.

\begin{algorithm}[!ht]
  \caption {Construction of the convolution layer}
  \textbf{Input: }
  \begin{algorithmic}
    \State - List of Ciphertext :
    \State \; $C_{in} := (C_{in(0)}, C_{in(1)}, \dots, C_{in(CH_{in}-1)}) \in (\mathcal{R}_{q}^{2})^{CH_{in}}$
    \State - Kernel: $Ker:= (Ker_{(o,i)})_{0 \leq o < CH_{out}, 0 \leq i < CH_{in}}$ 
    \State \qquad \qquad \qquad \qquad $\in \mathbb{R}^{CH_{out} \times CH_{in} \times K \times K}$
    \State - Bias: $B := (B_{(o)})_{0 \leq o < CH_{out}} \in \mathbb{R}^{CH_{out}}$
    \State - Interval: $I_{in}$
    \State - Stride: $S$ 
  \end{algorithmic}

  \textbf{Output: } 
  \begin{algorithmic}
    \State - List of Ciphertexts : $C_{out} \in (\mathcal{R}_{q}^{2})^{CH_{out}}$
    \State - Interval : $I_{out}$
  \end{algorithmic}

  \textbf{Procedure: } 
  *(All ciphertexts and operations are in $\mathcal{R}_{q}^{2}$)
    
    \begin{algorithmic}
        \For{$i=0$ to $CH_{in}-1$}
            \For{$j=0$ to $K-1$}
                \For{$k=0$ to $K-1$}
                    \State $C_{rot(i, j, k)}$ 
                    \State $\leftarrow Rot\left( C_{in(i)}, I_{in} \times (k + W_{img} \times j) \right)$
                \EndFor
            \EndFor
        \EndFor
        \For{$o=0$ to $CH_{out}-1$}
            \State $C_{out(t)}$ 
            \State $\leftarrow \sum_{i=0}^{CH_{in}-1} \sum_{j=0}^{K-1} \sum_{k=0}^{K-1} Mul \left(C_{rot(i,j,k)}, K_{(o,i,j,k)} \right)$
             
            \State $C_{out(o)} \leftarrow Add(C_{out(o)}, B_{(o)})$
        \EndFor
        \State $C_{out} \leftarrow \left(C_{out(0)}, \dots, C_{out(CH_{out}-1)}\right)$
        \State $I_{out} \leftarrow I_{in} \times S$ \\
        \Return{$C_{out}, I_{out}$}
    \end{algorithmic}
\label{algorithm:convolution layer}
\end{algorithm}

This approach improves the convolution process by decoupling the number of operations from the image size and limiting rotations to the square of the kernel size $K$.

The number of plaintext multiplications in the convolution layer of our method is $\mathcal{O}(CH_{in} \cdot CH_{out} \cdot K^{2})$ and the number of rotations is $\mathcal{O}(CH_{in} \cdot K^{2})$. Considering the complexity of plaintext multiplication $\mathcal{O}(N \cdot L)$ and rotation $\mathcal{O}(N \cdot \log N \cdot L^{2})$ as stated in~\cite{lee2023optimizations}, the overall complexity of the convolutional layer is $\mathcal{O}(N \cdot L \cdot CH_{in} \cdot K^{2} \cdot (CH_{out} + \log N \cdot L))$.


\subsection{Construction of the Average Pooling Layer}
\label{construction_of_the_average_pooling_layer}
CNNs use pooling layers to reduce input size, with max, min, and average pooling being common. Our architecture prioritizes high-speed CNN inference without bootstrapping. Ablation studies on a LeNet-5-like model revealed comparable accuracies for min (0.989), max (0.992), and average (0.985) pooling. In HE, max and min pooling require numerous multiplications~\cite{lee2023precise}, necessitating costly bootstrapping, while average pooling needs only one multiplication. Given the negligible performance difference and significant computational savings, we implemented average pooling in \sysname.

We optimize the average pooling layer by eliminating the constant multiplication by $1/c^2$ (where $c$ is the kernel size). Instead, we apply this multiplication in a preceding convolutional or flatten layer. This approach maintains functionality while reducing the number of multiplications and lowering depth, enhancing overall efficiency.

The implementation depends on the subsequent layer. If followed by an activation layer, the constant multiplication is moved to the preceding flatten or convolutional layer. For a linear layer $h(x) = Ax + b$, we use $h(x/c^2) = (1/c^2)Ax + b$, incorporating the scaling factor into the weight matrix. This strategy preserves the average pooling effect while optimizing computational resources.

When an activation function such as the square function or the approximate ReLU function follows average pooling, we cannot apply the above logic directly because these functions are not linear. In this case, we can apply the following logic:

\begin{itemize}[leftmargin=*]
\item If the activation function is the square function:
Then, $(1/c^2)^2 = 1/c^4$ is applied the following in the convolutional layer or the flatten layer.
\item If the activation function is the approximate ReLU function:
Then, we can apply the coefficient of the approximate ReLU ${f(x/c^2)}= 0.375373 + (0.5/c^2)x + (0.117071/c^4)x^{2}$ if the approximate ReLU is defined as ${f(x)}= 0.375373 + 0.5x + 0.117071x^{2}$.
\end{itemize}

The mechanism of average pooling closely resembles that of the convolutional layer. Suppose an average pooling operation is conducted with a kernel size of $c$. We can then apply a convolutional layer with $c$ as both the kernel size and stride and use a constant multiplication of $1/c^2$ as described. The overview of this logic is shown in Figure~\ref{figure:Conv2He}.

However, unlike the convolutional layer, \textit{invalid values} are introduced in the average pooling layer because the flattened kernels are not multiplied, as illustrated in Figure~\ref{figure:kernel_construction}. Due to this issue, the gap between data requires maximum rotation to the left. The exact interval is as follows:

\begin{align}
    (W_{img} + 1) \times (c-1)
\end{align}

\begin{figure}[!ht]
\centerline{\includegraphics[trim=0cm 0.1cm 0cm 0cm, clip=true, width=1\columnwidth]{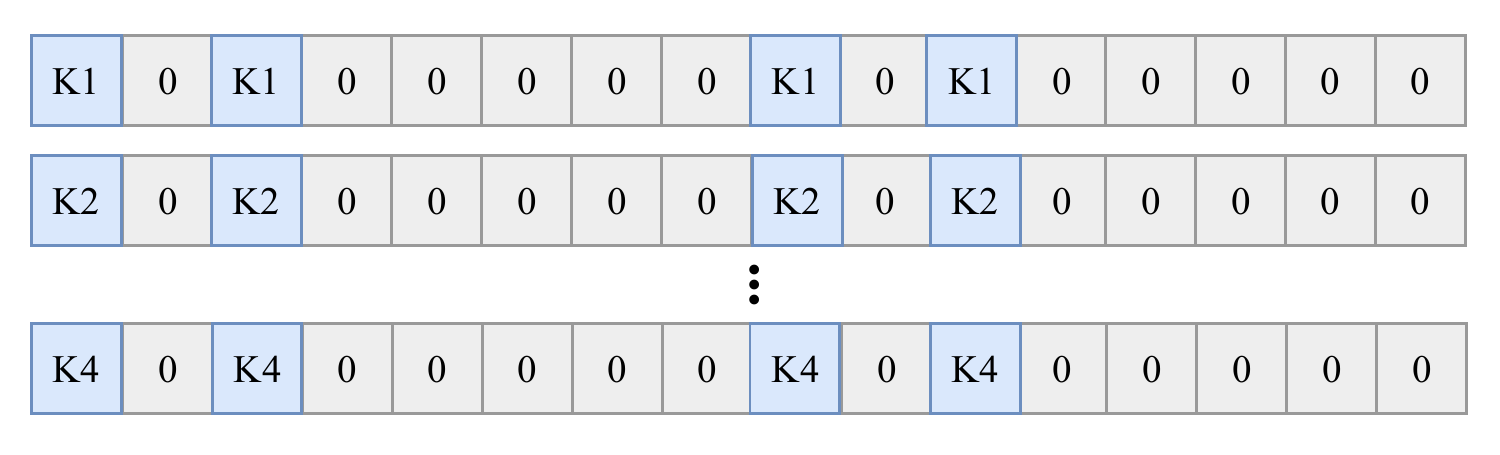}}
\caption{Flattening a two-dimensional kernel into a one-dimensional array.}
\label{figure:kernel_construction}
\end{figure}


\subsection{Construction of the Fully Connected (FC) Layer}
\label{Fully Connected layer converter} 

In an FC layer, the input data is multiplied by a weights matrix, as depicted in Figure~\ref{figure:fc_layer_mechanism} (a). PycrCNN performs the same procedure as a plain FC layer by multiplying encrypted data with the weights matrix. However, in this case, matrix multiplication requires significant computational power and time.


\begin{figure}[!ht]
    \centering
    \subfigure[PyCrCNN]{      \includegraphics[width=\linewidth]{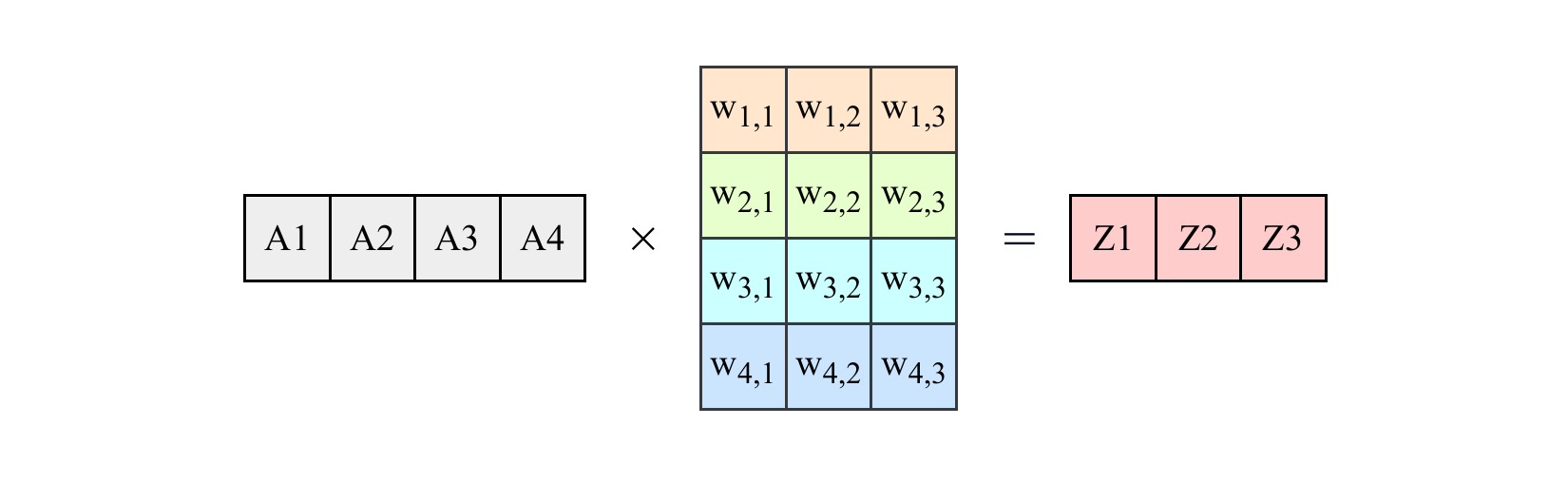}
        }
    \subfigure[TenSEAL]{      \includegraphics[width=\linewidth]{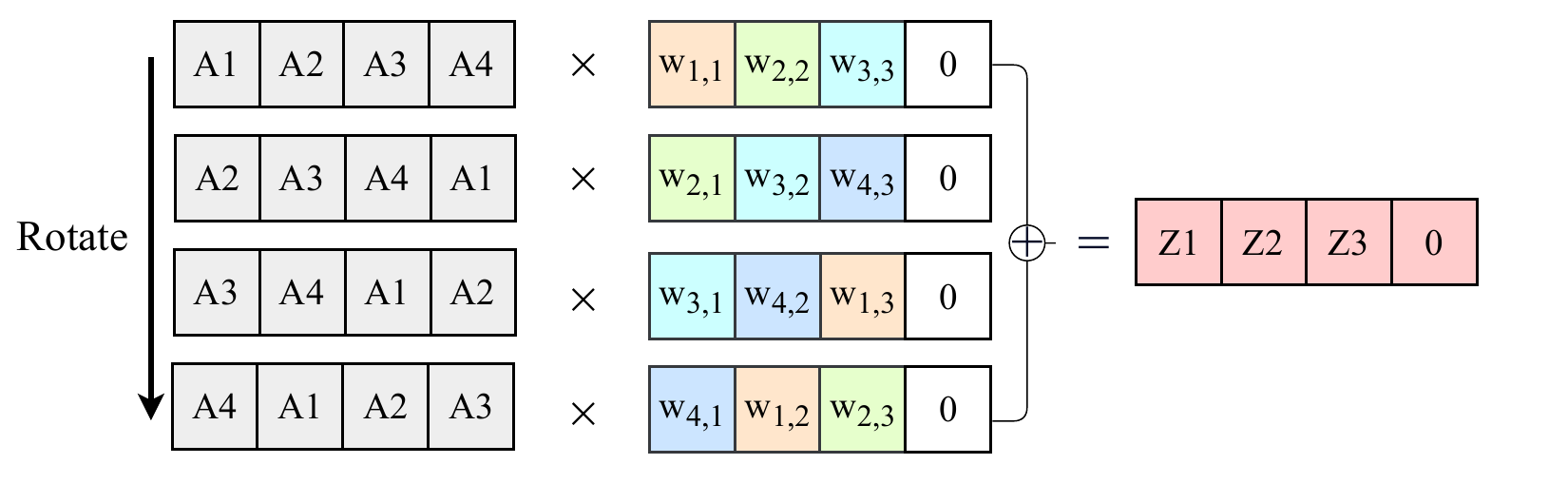}
        }
    \caption{Structure of the FC layer of (a) PyCrCNN and (b) TenSEAL.}
    \label{figure:fc_layer_mechanism}
\end{figure}

Employing the vector multiplication technique~\cite{halevi2014algorithms} offers a way to facilitate FC layer operations on the ciphertext. As depicted in Figure~\ref{figure:fc_layer_mechanism} (b), TenSEAL employs the diagonal method for vector-matrix multiplication~\cite{tenseal2021, halevi2014algorithms}. In this approach, rotated data is multiplied by a rotated weights vector. To utilize vector multiplication, the width of the weight matrix is set to the input size by padding with zeros. Consequently, the number of multiplications and rotations depends on the FC layer's input size. Generally, the input size of the FC layer is larger than its output size, which may consume unnecessary resources. To minimize the number of resource-intensive operations, including multiplication and rotation, we optimized the logic to depend on the FC layer's output size.



\begin{figure}[!ht]
    \centering
    \subfigure[Step 1]{      \includegraphics[width=0.4\linewidth]{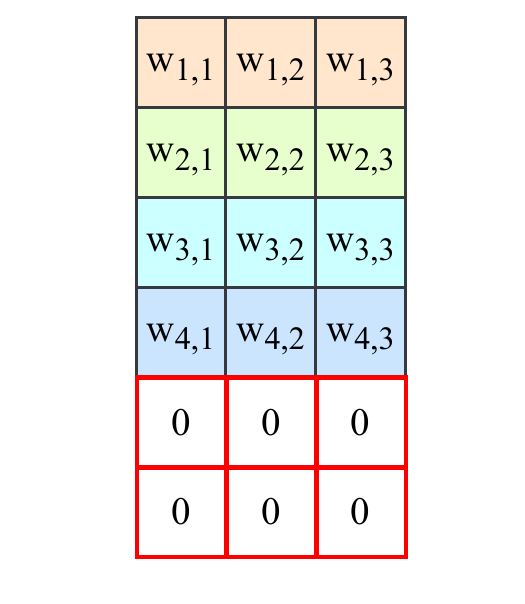}
        }
    \subfigure[Step 2]{      \includegraphics[width=0.4\linewidth]{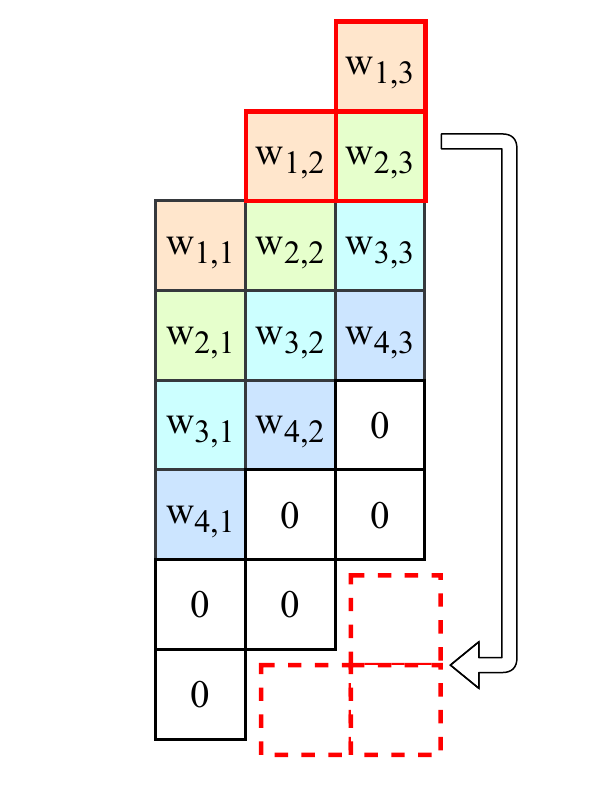}
        }
    \subfigure[Step 3]{      \includegraphics[width=0.4\linewidth]{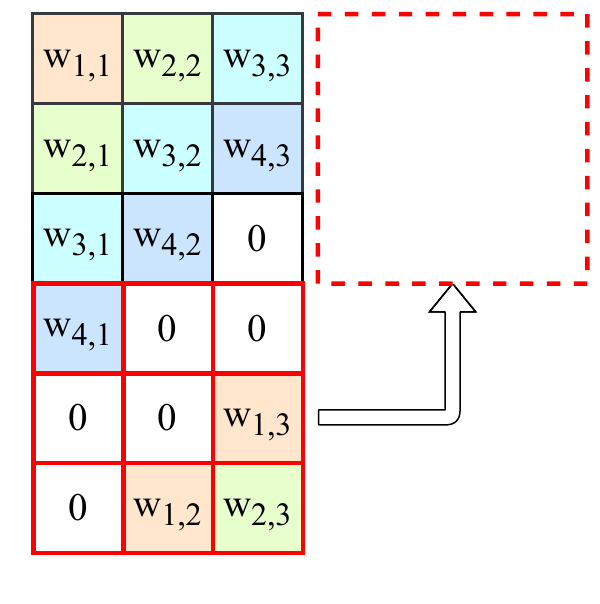}
        }
    \subfigure[Step 4]{      \includegraphics[width=0.4\linewidth]{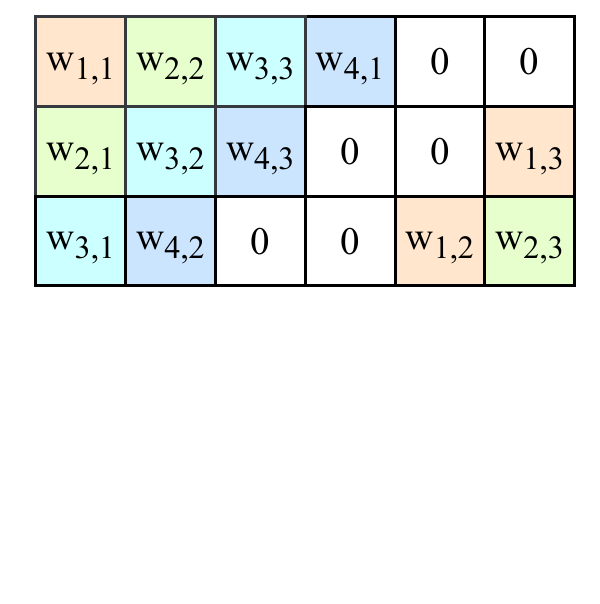}
        }
    \subfigure[Step 5]{      \includegraphics[width=\linewidth]{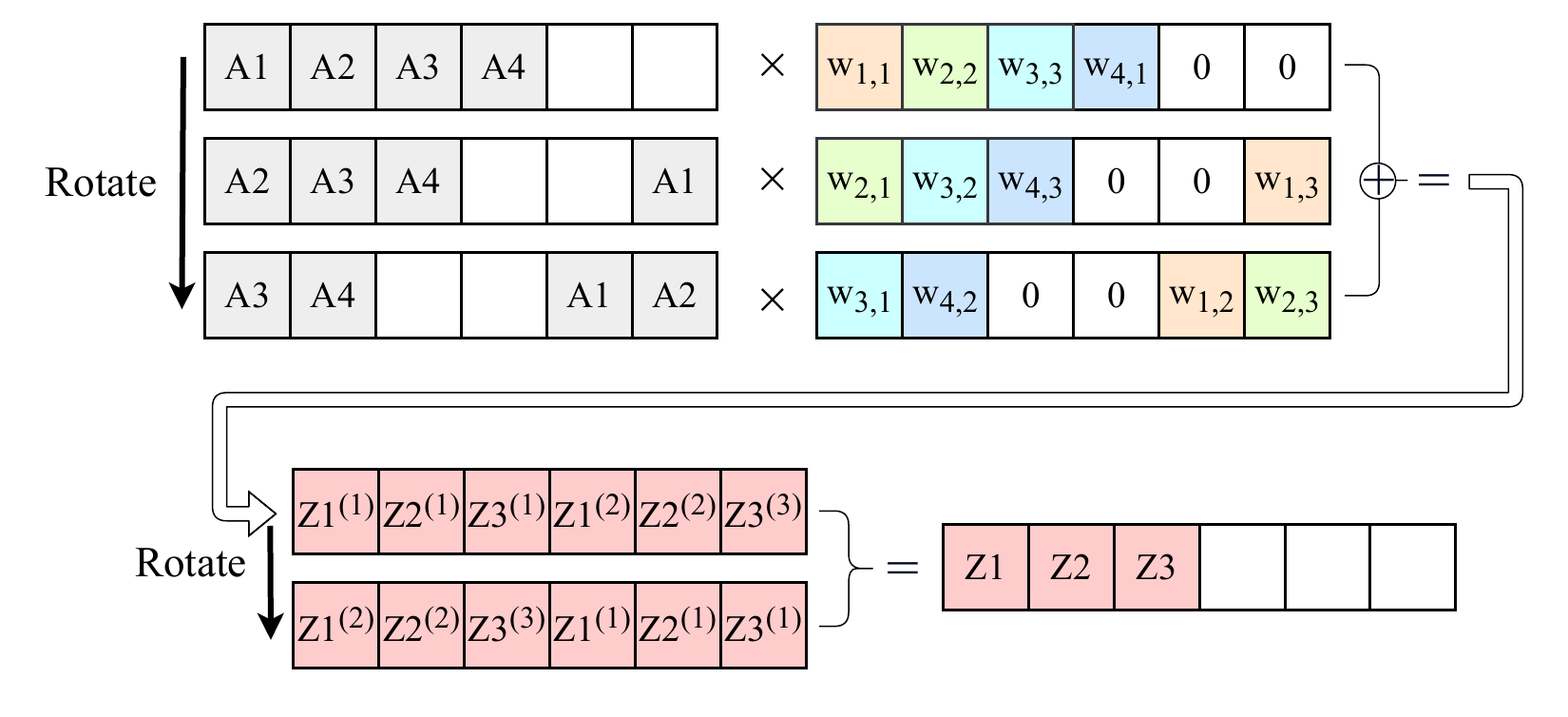}
        }
    \caption{Construction of FC layer of \sysname.}
    \label{figure:vec_multiplication}
\end{figure}

In Figure~\ref{figure:vec_multiplication}, the overall process of vector multiplication in \sysname is illustrated. First, as shown in Figure~\ref{figure:vec_multiplication} (a), pad zeros until the number of rows is a multiple of the output size, and then rotate the columns of the weight matrix upward, incrementing one column at a time as shown in Figure~\ref{figure:vec_multiplication} (b). Next, append the segments of each array with a negative index to the end of the array. Finally, as depicted in Figure~\ref{figure:vec_multiplication} (c), divide the matrix lengthwise by the output size and concatenate the parts horizontally; the weight matrix \( M_w \) then becomes as shown in Figure~\ref{figure:vec_multiplication} (d). Note that generating an optimized weight matrix (from Figure~\ref{figure:vec_multiplication} (a) to Figure~\ref{figure:vec_multiplication} (d)) involves modifying the weights in plaintext; however, the operation time is negligible.

In the inference process, the input ciphertext is rotated and multiplied by the weight matrix $M_w$, as depicted in Figure~\ref{figure:vec_multiplication} (e). Matrix multiplication can be performed in encrypted form with only multiplication operations that depend on the output size. Generally, since the output size of the FC layer is smaller than the input size, operations can be more efficient than with the diagonal method. For example, if the FC layer has an input size of 64 and an output size of 10, then the total number of multiplications and rotations would be 10, not 64.

\begin{algorithm}[!ht]
  \caption{Construction of the fully connected layer}
  \textbf{Input:}
  \begin{algorithmic}
    \State - Ciphertext: $C_{in} \in \mathcal{R}_{q}^{2}$
    \State - Weight matrix: $M_w \in \mathbb{R}^{DAT_{out} \times DAT_{in}}$
    \State - Bias: $B \in \mathbb{R}^{DAT_{out}}$
  \end{algorithmic}

  \textbf{Output:}
  \begin{algorithmic}
    \State - Ciphertext: $C_{out} \in \mathcal{R}_{q}^{2}$
  \end{algorithmic}

  \textbf{Procedure:}
  *(All ciphertexts and operations are in $\mathcal{R}_{q}^{2}$)
  \begin{algorithmic}
    \State $M_{rot} \leftarrow$ Apply $M_w$ from step 1 to 4 in Figure~\ref{figure:vec_multiplication}
    \State $W_{vec} = \lceil DAT_{in} / DAT_{out} \rceil \times DAT_{out}$
    \State $C_{sum} \leftarrow C(0)$
    \For{$o=0$ to $CH_{out}-1$}
      \State $V_{1} \leftarrow M_{rot(o)}[:W_{vec}-o] + [0] \times o$
      \State $C_{sum} \leftarrow Add(C_{sum}, Mul(Rot(C_{in}, o), V_{1}))$
      \State $V_{2} \leftarrow [0] \times (W_{vec}-o) + M_{rot(o)}[W_{vec}-o:]$
      \State $C_{sum} \leftarrow Add(C_{sum}, Mul(Rot(C_{in}, o-W_{vec}), V_{2}))$
    \EndFor
    \State $C_{out} \leftarrow C(0)$
    \For{$i=0$ to $\lceil DAT_{in}/DAT_{out} \rceil - 1$}
      \State $C_{out} \leftarrow Add(C_{out}, Rot(C_{sum}, i \times DAT_{out}))$
    \EndFor
    \State $C_{out} \leftarrow Add(C_{out}, B)$
    \Return{$C_{out}$}
  \end{algorithmic}
\label{algorithm:fully_connected_layer}
\end{algorithm}

The number of plaintext multiplications of the FC layer in \sysname is $\mathcal{O}(DAT_{out})$, and the number of rotations is $\mathcal{O}(DAT_{in}/DAT_{out} + DAT_{out})$. This means that if $DAT_{in}$ is smaller than the square of $DAT_{out}$, it can be denoted as $\mathcal{O}(DAT_{out})$. In this case, the time complexity of the FC layer depends only on the output size. Therefore, the complexity of the FC layer in \sysname is $\mathcal{O}((DAT_{in}/DAT_{out} + DAT_{out}) \cdot N \cdot \log N \cdot L^{2})$.

\subsection{Construction of the Flatten Layer}
\label{Flatten layer converter}


The outputs of the convolutional layer (or average pooling layer) typically have an interval (the reason is detailed in Section~\ref{Construction of the Convolutional Layer}). Additionally, the output of the convolutional layer is usually sparse, as the output from each filter resides in a distinct ciphertext. In a flatten layer, the operation is performed to collect each ciphertext and integrate it into a single ciphertext. Sparse ciphertext leads to unnecessary time and memory consumption. To avoid this, we construct flatten layers to remove invalid data between valid data and to compress all convolutional computation results into a single kernel ciphertext.

\subsubsection{Removing the Row Interval}
\label{The process of removing the row interval}

\begin{figure}[!th]
\centerline{\includegraphics[trim=0cm 0.1cm 0cm 0cm, clip=true, width=1\columnwidth]{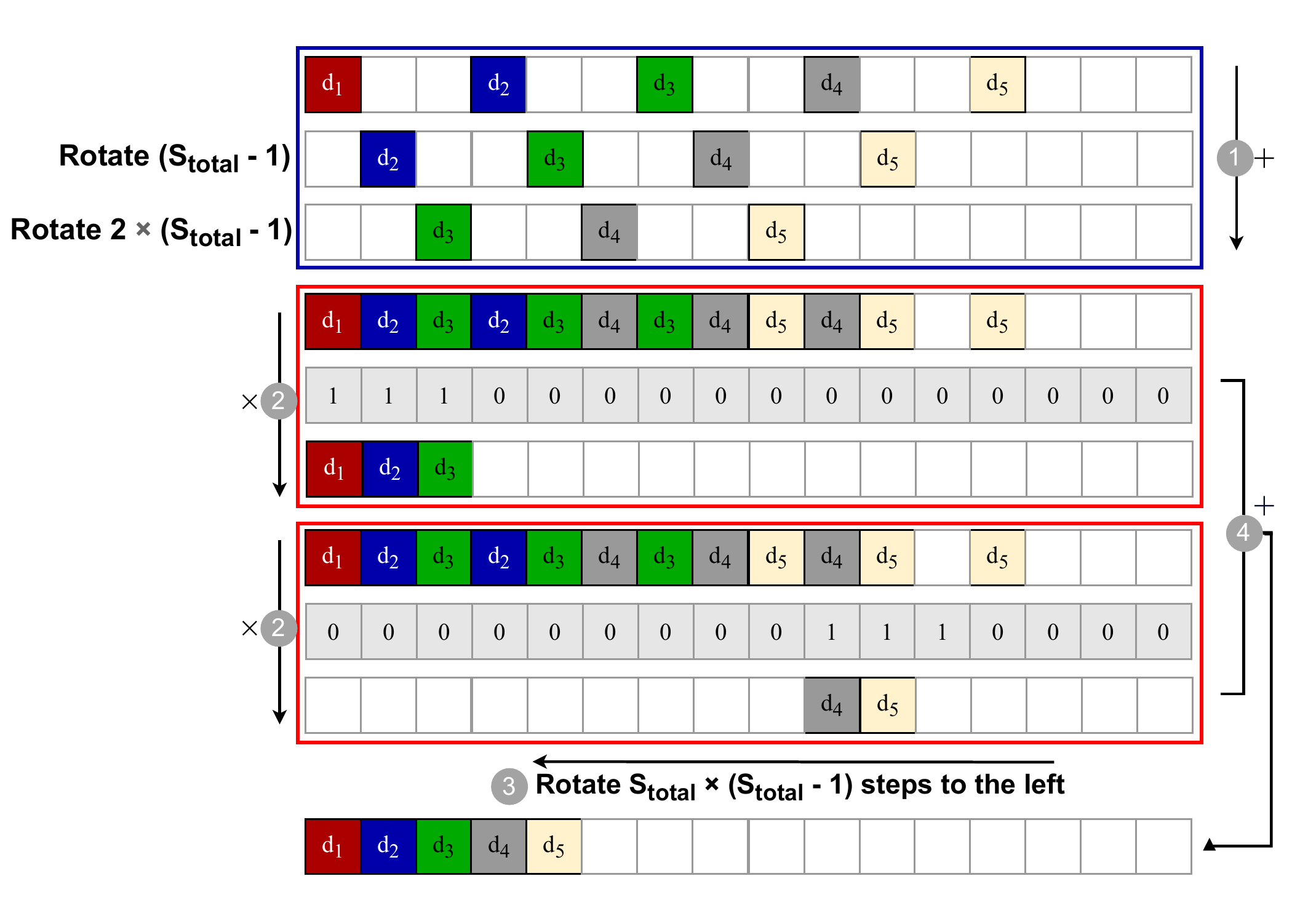}}
\caption{Removing row interval when the previous layer is either a convolutional layer or an approximate ReLU layer.}
\label{figure:flatten_layer3}
\end{figure}

After passing through the convolution layer or fully connected layer, the valid data in the output ciphertext will be offset, which we call row interval. We propose an algorithm to remove the row interval for two situations: whether the preceding layer is an average pooling layer or not. As detailed in Section~\ref{construction_of_the_average_pooling_layer}, the square layer does not handle constant multiplication, so it is not considered when identifying the previous layer.

If the previous layer is either a convolutional layer or an approximate ReLU layer, it may contain some zero values between each set of non-zero values. We can exploit this to efficiently execute the flatten operation by stacking the non-zero values using rotation and addition.

Additionally, if $W_{in} = 1$, there is only one column. Therefore, this removing row interval step can be skipped.

\begin{itemize}[leftmargin=*]
    \item In step {\footnotesize\CircledText{1}}, perform $(S_{total}-1)$ rotations $(S_{total}-1)$ times. Accumulate all the vectors produced from each rotation.
    \item In step {\footnotesize\CircledText{2}}, a vector containing $S_{total}$ ones and $(W_{img} - 1) \times S_{total}$ zeros is rotated in vector and multiplied $\lceil W_{in}/S_{total} \rceil$ times with the ciphertexts, allowing the extraction of $S$ data points at once.
    \item In step {\footnotesize\CircledText{3}}, after multiplications, the $i$th sparse vector is rotated $(i-1) \times S_{total} \times (S_{total} - 1)$ times to the left for $i \in \left[2, \lceil W_{in}/S_{total} \rceil \right]$.
    \item In step {\footnotesize\CircledText{4}}, all rotated vectors are added.
\end{itemize}


\subsubsection{Removing Invalid Data and Row Interval}~\label{sec:removing_pooling}

If the previous layer is an average pooling layer, \textit{invalid values} exist in the intervals between each set of valid data points. Moreover, the required constant multiplication (assuming a factor of $1/c^2$) has not been applied during the average pooling operation. Consequently, we must eliminate these invalid data points and perform the necessary multiplication. To address this issue, we suggest a three-step process as follows.

\begin{figure}[!th]
\centerline{\includegraphics[trim=0cm 0.1cm 0cm 0cm, clip=true, width=1\columnwidth]{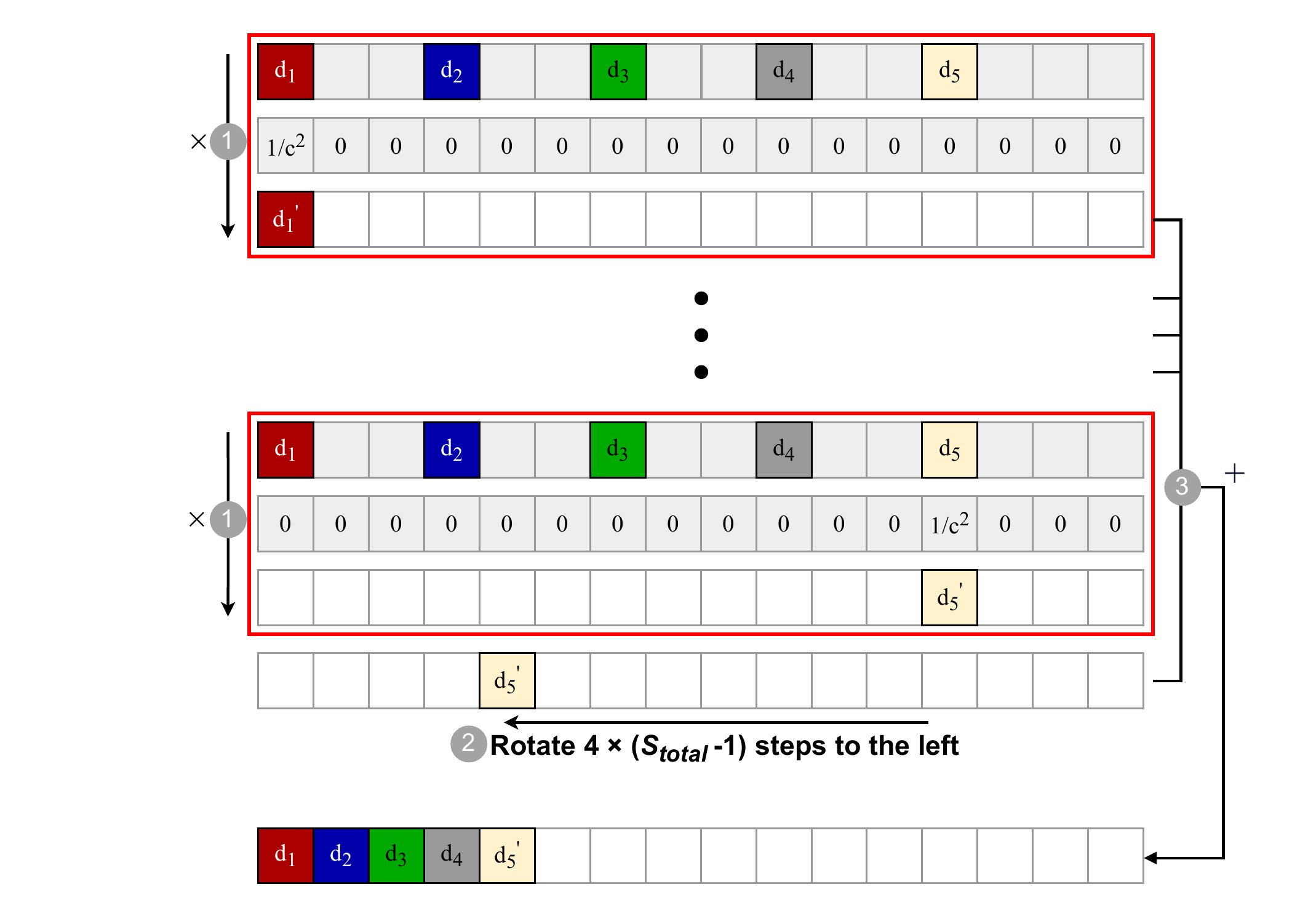}}
\caption{Removing invalid data between valid data when the previous layer is an average pooling layer.}
\label{figure:flatten_layer2}
\end{figure}

Figure~\ref{figure:flatten_layer2} {\footnotesize\CircledText{1}} illustrates the flatten layer's operational algorithm to remove invalid data. Unlike Figure ~\ref{figure:flatten_layer3}, some invalid values exist instead of the zero values when a flatten layer comes after an average pooling layer. Therefore, we can't operate data simultaneously. Thus, each value is needed to be calculated one by one.

Suppose $S_{total}$ is the value obtained by multiplying the stride values of all convolutional layers by the kernel sizes of all average pooling layers, $W_{img} \times H_{img}$ is the image size and $W_{in} \times H_{in}$ is the output size of the last average pooling layer.

\begin{itemize}[leftmargin=*]
    \item In step {\footnotesize\CircledText{1}}, a vector containing one $1/c^2$ and $(W_{img} \times S_{total}-1)$ zeros is rotated in vector and multiplied $W_{in}$ times with the ciphertexts.
    \item In step {\footnotesize\CircledText{2}}, after multiplication, the $i$-th sparse vector is rotated $(i-1) \times (S_{total} - 1)$ times to the left for $i \in \left[ 2, W_{in}  \right]$.
    \item In step {\footnotesize\CircledText{3}}, all rotated vectors are added.
\end{itemize}

\subsubsection{Removing the Column Interval}
\label{The process of removing the column interval}

\begin{figure}[!ht]
\centerline{\includegraphics[trim=0cm 0.1cm 0cm 0cm, clip=true, width=1\columnwidth]{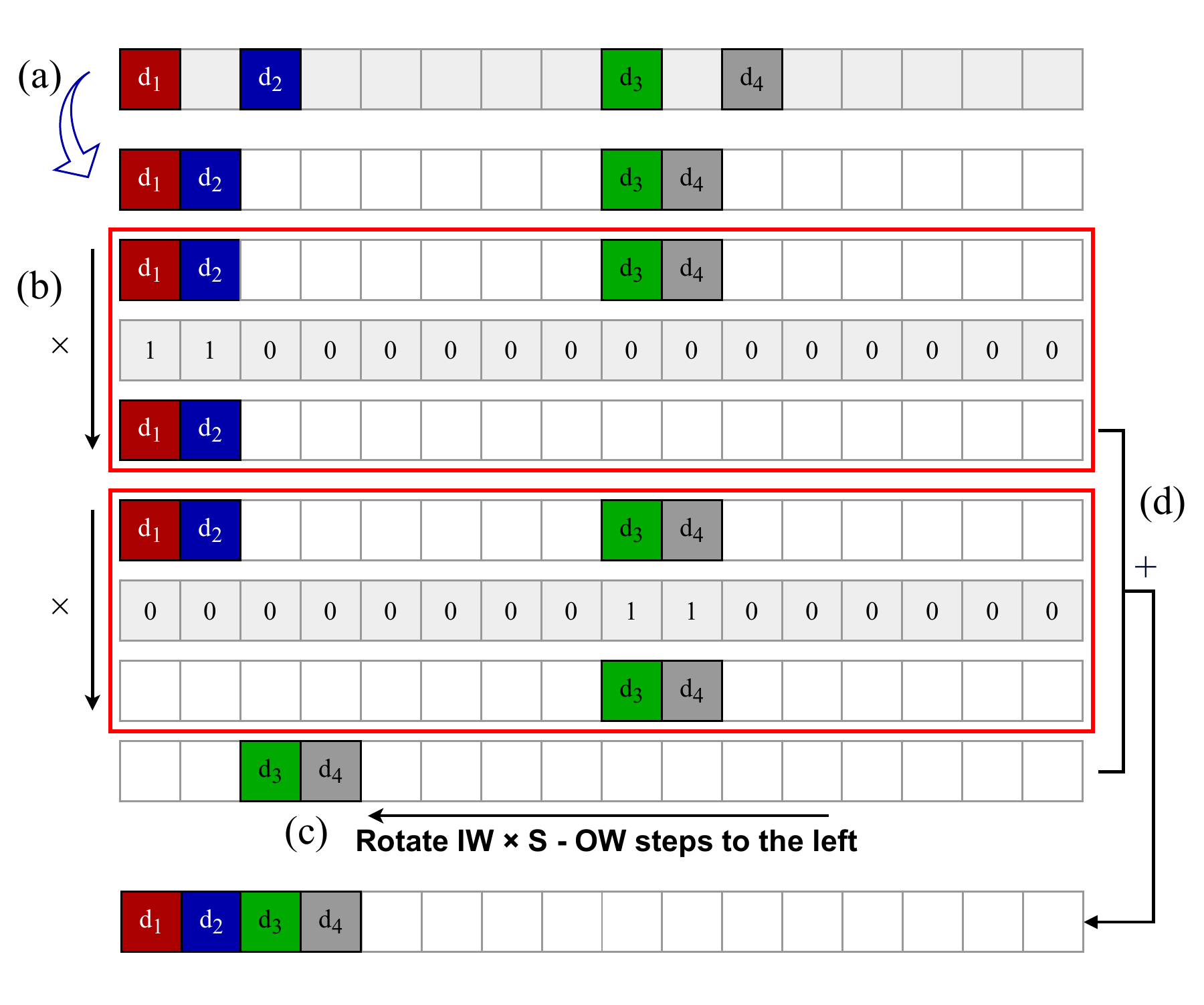}}
\caption{Overall mechanism for compressing all convolutional computation results into a single kernel ciphertext.}
\label{figure:flatten_layer1}
\end{figure}

As shown in Figure~\ref{figure:flatten_layer1}, the overall computation mechanism within the HE-based convolutional layer comprises four steps. 
After removing the row interval (step {\footnotesize\CircledText{1}}), it is apparent that the column interval still exists. For models utilizing 1D CNN or 2D CNN inference with $H_{in} = 1$, this removing column interval step can be skipped due to the presence of only one row. However, for 2D CNN models, an additional set of 3 steps is executed to remove the column interval, as indicated in Figure~\ref{figure:flatten_layer1}.

\begin{itemize}[leftmargin=*]
    \item In step {\footnotesize\CircledText{2}}, a vector containing $W_{img}$ ones and the rest as zeros is rotated and multiplied with the ciphertexts.
    \item In step {\footnotesize\CircledText{3}}, after the multiplication, rotate the $i$-th sparse vector $(i-1) \times (W_{img} \times S_{total} - W_{in}) $ times to the left, where $i \in \left[ 2, H_{in} \right]$.
    \item In step {\footnotesize\CircledText{4}}, sum up all the rotated vectors.
\end{itemize}

These additional steps help to remove the column intervals effectively, ensuring that the data is adequately flattened and ready for further processing in the CNN pipeline.

\subsubsection{Packing Multiple Ciphertexts}
\begin{figure}[!ht]
\centerline{\includegraphics[trim=0cm 0.1cm 0cm 0cm, clip=true, width=1\columnwidth]{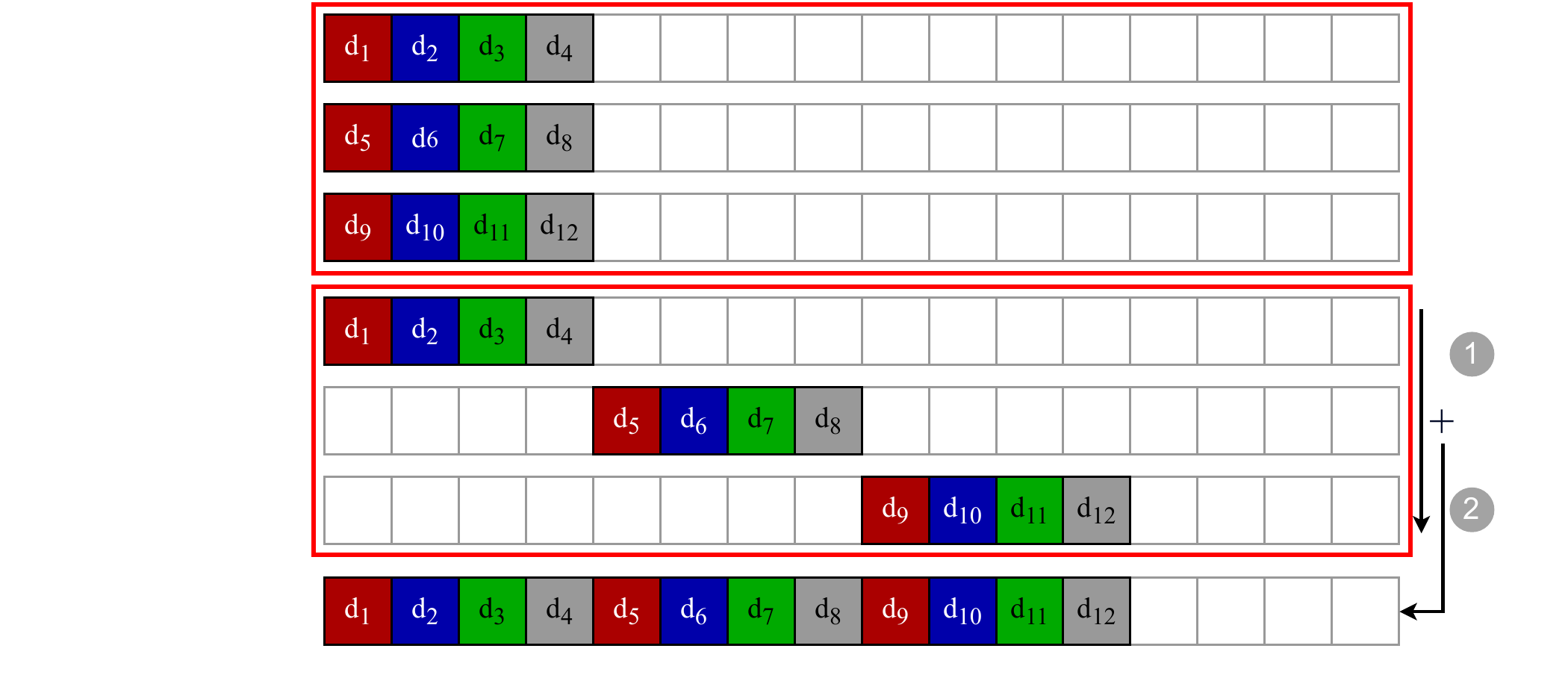}}
\caption{Overall mechanism for multiple ciphertexts into a single ciphertext.}
\label{figure:flatten_layer4}
\end{figure}
Figure~\ref{figure:flatten_layer4} shows the detail of the steps taken to combine the multiple output vectors produced by the convolution operation into a singular ciphertext. The descriptions are as follows:

\begin{itemize}[leftmargin=*]
\item In step {\footnotesize\CircledText{1}}, each of the output vectors is right-rotated $W_{in} \times H_{in}$ steps.
\item In step {\footnotesize\CircledText{2}}, all these rotated vectors are aggregated together. This aggregation step is particularly efficient because it avoids the need for multiplication operations, which are computationally intensive when applied to ciphertexts. This efficiency is due to the sparsity of the vector.
\end{itemize}

After completing this procedure, the outputs are flattened into a single ciphertext. The number of slots occupied in this ciphertext will be $W_{in} \times H_{in} \times CH_{in}$.

The convolutional layer will also require two multiplication operations when the stride $S$ exceeds one. However, if $S_{total}$ equals one, the process illustrated in Figure~\ref{figure:flatten_layer1} (a) becomes unnecessary. Thus, only a single multiplication operation would be needed. 

\subsubsection{Time Complexity}

The number of plaintext multiplications in the flatten layer of \sysname is $\mathcal{O}(W_{in} + H_{in})$ and the number of rotations in the flatten layer of \sysname is $\mathcal{O}(CH_{in} \cdot S_{total} + W_{in} + H_{in})$.

Therefore the complexity of flatten layer of \sysname is $\mathcal{O}((CH_{in} \cdot S_{total} + W_{in} + H_{in}) \cdot N \cdot \log N \cdot L^{2})$

  


\subsection{Construction of the Activation Layer} 

Activation functions require many multiplications~\cite{lee2023precise}, so it requires many computing resources to implement them with HE. To address this, we approximate each activation function using a low-degree polynomial. For instance, a 2-degree polynomial approximation of the ReLU function~\cite{HE_Accurate_CNN} is represented as ${f(x) = 0.375373 + 0.5x + 0.117071x^2}$. In \sysname, these approximated activation functions are used instead of the real activation functions for efficiency.

\section{Optimizing Ciphertext Size for Batch Processing} 

As shown in Figure~\ref{encryption_for_batch_data}, \sysname combines multiple encrypted input data into a single ciphertext for batch processing. However, the ciphertext size for each input should not be allocated solely based on the size of the input data. Instead, it should also consider the intermediate or final output values produced by performing CNN operations. This prevents overlapping issues with other encrypted input data's intermediate or final output values. We can determine this size in advance by examining the structure of the CNN model. 

This section will explain how the output size can be determined in each of the layers (convolutional layer, average pooling layer, flatten layer, and fully connected layer) implemented in \sysname.

\subsection{Convolutional Layer}

The size of the convolution layer's output data depends on three parameters: stride, padding, and kernel size.

\begin{itemize}[leftmargin=*]
    \item Stride: In the convolutional layer of \sysname, the interval between each row of data equals the product of the strides from all previous layers, while the interval between each column equals this product multiplied by the image width.
    
    \item Padding: In a convolutional layer with padding, extra space is required to add 0 elements. If there are $L$ convolutional layers with padding spaces $P_1, P_2, ..., P_L$, then allocate a space of length $P = P_1 + P_2 + ... + P_L$ in the input ciphertext. This space is used for each padding process. As padding is accounted for in the initial input image size, it is not considered when calculating the largest size in \textbf{Theorem} \ref{convolution_do_not_increse}.
    
    \item Kernel size: A larger kernel size results in a smaller output size, as explained in \textbf{Theorem}~\ref{convolution_do_not_increse}.    
\end{itemize}

The output size of the convolutional layer does not always exceed the input size when the kernel is greater than or equal to the stride. This is because the convolutional operation is designed to be calculated independently in each channel. The proof of this statement is described in \textbf{Theorem}~\ref{convolution_do_not_increse}. 

In \textbf{Theorem}~\ref{convolution_do_not_increse}, one condition is that the stride value is not larger than the kernel size. This is realistic, as a larger stride value would result in data loss from the image.


    
\begin{theorem}
\label{convolution_do_not_increse}
    Let $W_{img} \times H_{img}$ be the image size and $N_{layer}$ the number of total convolutional layers. Denote $(W_{in(i)}, H_{in(i)})$ and $(W_{out(i)}, H_{out(i)})$ as the input and output sizes of the $i$-th convolutional layer, and the kernel size is $(W_{ker(i)}, H_{ker(i)})$, stride size is $(W_{st(i)}, H_{st(i)})$, and padding size is $(W_{pad(i)}, H_{pad(i)})$ for all $i \in \left[1, N_{layer}\right]$. Additionally, let
    
    \notag
    \begin{align}
        H_{[st(i)]} = \prod\limits_{j=1}^i H_{st(j)}
    \end{align}

    Then, $H_{img} \geq H_{out(i)} \times H_{[st(i)]}$, where $H_{ker(i)} \geq H_{st(i)} \geq 1$ for all $i \in \left[1, N_{layer}\right]$.
\end{theorem}

\begin{proof}
    We obtain that $H_{out(i)}$ :
    \notag
    \begin{align}
    H_{out(i)} = \left[\frac{H_{in(i)} - H_{ker(i)}}{H_{st(i)}} \right] + 1
    \end{align}

    Given the hypothesis, the following inequality holds:
    
    \begin{align}
        H_{out(i)} \times H_{st(i)} &= \left(\left[\frac{H_{in(i)} - H_{ker(i)}}{H_{st(i)}} \right] + 1 \right) \times H_{st(i)} \\
        &= \left[\frac{H_{in(i)} - H_{ker(i)} + H_{st(i)}}{H_{st(i)}} \right] \times H_{st(i)} \\
        &\leq \left[\frac{H_{in(i)}}{H_{st(i)}} \right] \times H_{st(i)} \leq H_{in(i)} = H_{out(i-1)}
    \end{align}

When $i > 1$, $H_{in(i)} = H_{out(i-1)}$.
As proven by mathematical induction, inserting an average pooling layer or an activation function between convolutional layers does not change the size. 
\end{proof}

Note that after passing through the $i$-th convolutional layer, the interval in each column of data is $W_{img} \times (H_{[st(i)]} - 1)$ (detailed in ~\ref{Construction of the Convolutional Layer}).
Given $W_{img}$ columns of data or fewer, the image size of $H_{img} \times W_{img}$ is not exceeded by \textbf{Theorem}~\ref{convolution_do_not_increse}.

\begin{corollary}
\label{convolution_row_do_not_increse}
    Denote
    \notag
    \begin{align}
        W_{[st(i)]} = \prod\limits_{j=1}^i W_{st(j)}
    \end{align}
    Then, $W_{img} \geq W_{out(i)} \times W_{[st(i)]}$, where $W_{ker(i)} \geq W_{st(i)} \geq 1$ for all $i \in \left[1, N_{layer}\right]$.
\end{corollary}

\begin{proof}
The proof is analogous to \textbf{Theorem}~\ref{convolution_do_not_increse} and is therefore omitted.
\end{proof}

We adopt a slot-based operation approach, allowing data with remaining intervals to pass through to the next layer unchanged. The \textit{row interval} is defined as the length between row vectors, and the \textit{column interval} is the length between column vectors. Specifically, if the stride of the $i$-th convolutional layer is $S_i$, the interval between each data value after this layer is:

\notag
\begin{align}
    \text{row interval} &: \prod\limits_{j=1}^{i} S_j - 1 \\ 
    \text{column interval} &: W_{img} \times (\prod\limits_{j=1}^{i} S_j - 1)
\end{align}

By Theorem~\ref{convolution_do_not_increse} and Corollary~\ref{convolution_row_do_not_increse}, each data value stays within the specified index, confirming the flawless operation of the convolutional layer in \sysname.



\subsection{Average Pooling Layer}
Assuming that the kernel size in the average pooling layer is $c$, each channel uses a filter with a kernel size of $(c, c)$ and a stride of $(c, c)$. The result is then multiplied by $1/c^2$. We omit the multiplication of a vector consisting of 0s and 1s, allowing some non-zero but unused values (i.e., \textit{invalid values}) to overlap with other values. According to \textbf{Theorem}~\ref{convolution_do_not_increse}, the number of slots in the ciphertext data remains unchanged when passing through \sysname's average pooling. However, to prevent \textit{invalid values} from occupying the positions of other used data, an additional space of $(W_{img}+1) \times (c-1)$ is required (detailed in \ref{construction_of_the_average_pooling_layer}).

\subsection{Flatten Layer}
Assume that we flatten data from $CH_{in}$ channels, each of size $W_{in} \times H_{in}$. After flattening, the ciphertext data will have $W_{in} \times H_{in} \times CH_{in}$ slots. This size can exceed the maximum size of the layers preceding the flatten layer. We compare these two sizes to determine the largest possible size of the hidden layer. Because the model has plaintext information, comparison operations are possible.

\subsection{Fully Connected Layer}

Let $CH_{in}$ and $CH_{out}$ be the input size and output size of the FC layer, respectively. The size of the slot required in the FC layer is given by:

\notag
\begin{align}
    CH_{out} \times \left \lceil \frac{CH_{in}}{CH_{out}} \right \rceil 
\end{align}

Details of how this formula was derived can be found in Section~\ref{Fully Connected layer converter}. The maximum size up to the preceding layer is compared with this value and updated to the larger of the two.
\section{Experiments}
\label{sec:exmeriments_and_results}

This section presents experimental results demonstrating the feasibility and effectiveness of \sysname. 

We implement seven distinct CNN models, including LeNet-1~\cite{lecun1989handwritten} and LeNet-5~\cite{lecun1998gradient}, to demonstrate that \sysname can be applicable to a wide range of model architectures.

\subsection{Experimental Settings}

We used a server on the NAVER Cloud platform~\cite{navercloudServer} with the following specifications: 16 vCPU cores (Intel(R) Xeon(R) Gold 5220 CPU @ 2.20GHz), 64GB of memory, and a 50GB SSD.

\subsection{Datasets}

We utilized four different datasets for our models: MNIST~\cite{mnist}, CIFAR-10~\cite{cifar10}, USPS~\cite{uspsdataset}, and ECG~\cite{ecgdataset}.

\subsubsection{MNIST Dataset}
The MNIST dataset~\cite{mnist} consists of 70,000 grayscale images of handwritten digits, each sized $28 \times 28$ pixels, with 10 classes representing the digits from 0 to 9. We used 60,000 images for training and 10,000 images for testing.

\subsubsection{CIFAR-10 Dataset}
The CIFAR-10 dataset~\cite{cifar10} includes 60,000 RGB images, each sized $32 \times 32$ pixels, with 10 classes. We used 50,000 images for training and 10,000 images for testing.

\subsubsection{USPS Dataset}
The USPS dataset~\cite{uspsdataset} contains 9,298 grayscale images of handwritten digits, each sized $16 \times 16$ pixels, with 10 classes. We used 7,291 images for training and 2,007 images for testing.

\subsubsection{ECG Dataset}
The ECG dataset from the MIT-BIH arrhythmia database~\cite{ecgdataset} consists of 109,446 samples at a sampling frequency of 125 Hz. It includes signals representing ECG shapes of heartbeats in normal cases and various arrhythmias. Following the preprocessing steps by Abuadbba et al.~\cite{abuadbba2020use}, we used a subset of 26,490 samples, with 13,245 samples for training and 13,245 samples for testing, available at \url{https://github.com/SharifAbuadbba/split-learning-1D}.

\subsection{Hyperparameter Settings for Training}

We trained all models with a learning rate of 0.001 and a batch size of 32. The MNIST and USPS models were trained for 15 epochs, the CIFAR-10 model for 25 epochs, and the ECG model for 30 epochs, based on their convergence behavior.

\subsection{Library and Parameter Setup}

We used SEAL-Python~\cite{huelseseal} for homomorphic operations, ensuring consistency with our baseline models TenSEAL and PyCrCNN, which are also SEAL-based. PyTorch and torchvision were used for model training. Detailed information about libraries and the experimental setup is available at our GitHub repository: \url{https://github.com/hm-choi/uni-henn}. Table~\ref{table:parameter_setting} provides the specific parameters used in our experiments. All security parameters comply with the 128-bit security level as described in~\cite{rahulamathavan2022privacy}.

\begin{center}
\begin{table}[!ht]
\setlength{\tabcolsep}{4.5pt}
\centering
\caption{Parameters used in \sysname. $PK$: public key size (encryption key), $SK$: secret key size (decryption key), $GK$: Galois key size, $RK$: Relinearization keys size. $GK$ and $RK$ together form the evaluation key. \# mult: total allowed multiplications.}
\begin{tabular}{|c|c|}
\noalign{\smallskip}\noalign{\smallskip}\hline
\# slots & 8,192 \\
\hline\hline
scale factor & 32 \\
\hline
log Q & 432 \\
\hline 
$PK$ (MB) & 1.87 \\
\hline
$SK$ (MB) & 0.94 \\
\hline
$GK$ (GB) & 0.57 \\
\hline
$RK$ (MB) & 22.52 \\
\hline
$ctxt$ (MB) & 1.68 \\
\hline
\# mult (\textit{depth}) & 11 \\
\hline
\end{tabular}
\label{table:parameter_setting}
\end{table}
\end{center} 

\section{Experimental Results}
\label{experiment}
In this section, we perform various experiments to verify the feasibility and evaluate the performance of \sysname. Each experiment in this paper was conducted under identical conditions and repeated 30 times to ensure consistency. The objectives of our experiments are as follows:

\begin{itemize}[leftmargin=*]
\item \textbf{Comparison of Inference Time with State-of-the-Art Solutions.} In this section, we compare the inference performance of \sysname against two state-of-the-art HE-based deep learning inference frameworks: TenSEAL~\cite{tenseal2021} and PyCrCNN~\cite{disabato2020privacy}. We selected TenSEAL for comparison because it is a well-known open-source library that uses the \texttt{im2col} algorithm to implement CNNs. This algorithm allows TenSEAL to achieve highly optimized and efficient inference times for convolution operations. However, this efficiency comes at the cost of flexibility: TenSEAL's architecture only supports models with a single convolutional layer, making it unsuitable for CNN architectures with multiple convolutional layers. PyCrCNN, on the other hand, does not employ the \texttt{im2col} algorithm, making it applicable to a broader range of CNN architectures. However, this flexibility results in slower inference times compared to TenSEAL. Our experiments show that \sysname successfully balances both flexibility and efficiency. It achieves inference times comparable to those of TenSEAL while supporting CNN architectures with multiple convolutional layers, making \sysname a more versatile solution for HE-based deep learning inference.

\item \textbf{Adaptability of \sysname for Various CNN Model Architectures.} One of the compelling features of \sysname is its ability to adapt to a variety of CNN architectures, including both complex and 1D CNN models. To demonstrate this adaptability, we conducted experiments with several CNN models, including LeNet-5, a seminal CNN model widely adopted for digit recognition tasks~\cite{zhang2023adding, liu2022convnet}. LeNet-5 is more complex than many other models, with multiple convolutional layers and fully connected layers. Our successful implementation of LeNet-5 in \sysname highlights the system's ability to handle large and complex CNN models effectively. We also considered a 1D CNN architecture particularly useful for sequence-based tasks such as time-series analysis, natural language processing, and certain bioinformatics applications. This capability sets \sysname apart from solutions like TenSEAL, which is constrained to supporting only single convolutional layer models and does not offer support for 1D CNNs. Through these experiments, we aim to demonstrate \sysname's comprehensive applicability and its ability to adapt to various CNN architectures, making it a versatile tool for secure and efficient deep learning inference across diverse application domains.

\item \textbf{Adaptability of \sysname with Various Datasets.} One of the key strengths of \sysname is its adaptability to various kinds of data. To demonstrate this, we conducted experiments using four diverse datasets: MNIST, CIFAR-10, USPS, and ECG. MNIST and CIFAR-10 are widely used image classification datasets that serve as standard benchmarks in the deep learning community. MNIST consists of grayscale images of handwritten digits, while CIFAR-10 comprises coloured images spanning ten different object classes. USPS is another grayscale image dataset used for handwriting recognition. ECG represents electrocardiogram data and is commonly used in healthcare applications for diagnosing various heart conditions. By demonstrating that \sysname performs well across these varied datasets, we aim to show its wide applicability to tasks involving images of different sizes and complexities, as well as specialized domains like healthcare.
\end{itemize}

\subsection{Experiment 1: Comparison of Inference Time between \sysname, \tenseal, and \pycrcnn}
\label{experiment1}

In this experiment, we implement the model, denoted as $M_1$, using \sysname, TenSEAL, and PyCrCNN. The $M_1$ model was introduced in the TenSEAL paper~\cite{tenseal2021} and is well-suited for implementation using TenSEAL. Despite its simple architecture, which contains only one convolutional layer, the model achieves a high accuracy of 97.65\%. Table~\ref{table:parameter_setting_M1} provides detailed specifications for this model architecture. The Conv2D parameters -- $CH_{in}$, $CH_{out}$, $K$, and $S$ -- represent the number of input channels, the number of output channels, kernel size, and stride for a two-dimensional convolutional layer, respectively. Similarly, $DAT_{in}$ and $DAT_{out}$ in the FC layer indicate the input and output dimensions. We use these notations to represent all other remaining models presented in this paper.


\begin{table}[!ht]
\begin{center}
\caption{Detailed parameters for $M_1$.}
\renewcommand{\arraystretch}{1.0}
\begin{tabular}{|c|c|c|}
\hline
Layer & Parameter & \# Mult\\
\hline\hline
Conv2d & $CH_{in}$ = 1, $CH_{out}$ = 8, $K$ = 4, $S$ = 3 & 1\\
\hline
Square & - & 1 \\
\hline
Flatten & - & 2\\
\hline
FC1  & $DAT_{in}$ = 648, $DAT_{out}$ = 64 & 1 \\
\hline
Square  & - & 1 \\
\hline
FC2  & $DAT_{in}$ = 64, $DAT_{out}$ = 10 & 1 \\
\hline
Total \#Mult & - & 7 \\
\hline
\end{tabular}
\label{table:parameter_setting_M1}
\end{center}
\end{table}

Table~\ref{table:inference_experiment1} presents the results. In terms of inference time for a single sample, the TenSEAL implementation exhibited the highest performance, taking 6.298 seconds. TenSEAL achieved this efficiency by optimizing convolution operations on the input data using the \texttt{im2col} algorithm. \sysname took 16.247 seconds, making it 2.6 times slower than TenSEAL. PyCrCNN was the slowest, requiring 154.494 seconds. However, \sysname has the advantage when handling multiple data points simultaneously. It outperforms TenSEAL by supporting batch operations for up to ten input samples, allowing concurrent inference calculations. TenSEAL and PyCrCNN do not offer parallel processing for multiple samples, resulting in total times that increase proportionally with the number of input samples. To process ten input samples together, \sysname still takes 16.247 seconds, approximately 3.9 times and 94.6 times faster than TenSEAL's 63.706 seconds and PyCrCNN's 1537.227 seconds, respectively. 


\begin{table}[!ht]
\begin{center}
\caption{Comparison of the average (with standard deviation) inference time in seconds on the MNIST dataset between PyCrCNN and \sysname for the $M_1$ model architecture.}
\renewcommand{\arraystretch}{1.0}
\resizebox{0.99\linewidth}{!}{
\begin{tabular}{|c|c|c|c|c|c|}
\hline
Layer & \multicolumn{2}{c|}{TenSEAL}&  \multicolumn{2}{c|}{PyCrCNN} & \sysname\\
\hline\hline
\# of samples & 1 & 10 & 1 & 10 & 1 \& 10 \\
\hline
Drop Level & - & - & - & - & \begin{tabular}[c]{@{}c@{}} 0.065 \\ (0.002)\end{tabular} \\
\hline
Conv2d & \begin{tabular}[c]{@{}c@{}} 2.522 \\ (0.041)\end{tabular} &\begin{tabular}[c]{@{}c@{}}25.202   \\(0.121)\end{tabular} & \begin{tabular}[c]{@{}c@{}} 31.844 \\ (0.483)\end{tabular} & \begin{tabular}[c]{@{}c@{}} 316.536 \\ (0.847)\end{tabular} & \begin{tabular}[c]{@{}c@{}} 3.436 \\ (0.036)\end{tabular} \\
\hline
Square & \begin{tabular}[c]{@{}c@{}} 0.024 \\ (0.001)\end{tabular} & \begin{tabular}[c]{@{}c@{}} 0.242 \\ (0.003)\end{tabular} & \begin{tabular}[c]{@{}c@{}} 25.082 \\ (0.548)\end{tabular} & \begin{tabular}[c]{@{}c@{}} 247.291 \\ (0.844)\end{tabular} & \begin{tabular}[c]{@{}c@{}} 0.281 \\ (0.009)\end{tabular} \\
\hline
Flatten & - & - & \begin{tabular}[c]{@{}c@{}} 0.000 \\ (0.000)\end{tabular} & \begin{tabular}[c]{@{}c@{}} 0.000 \\ (0.000)\end{tabular} & \begin{tabular}[c]{@{}c@{}} 5.285 \\ (0.060)\end{tabular} \\
\hline
FC1  & \begin{tabular}[c]{@{}c@{}} 3.540 \\ (0.293)\end{tabular} & \begin{tabular}[c]{@{}c@{}} 36.147 \\ (1.084)\end{tabular} & \begin{tabular}[c]{@{}c@{}} 94.883 \\ (0.287)\end{tabular} & \begin{tabular}[c]{@{}c@{}} 946.841 \\ (1.407)\end{tabular} & \begin{tabular}[c]{@{}c@{}} 6.783 \\ (0.064)\end{tabular} \\
\hline
Square  & \begin{tabular}[c]{@{}c@{}} 0.014 \\ (0.001)\end{tabular} & \begin{tabular}[c]{@{}c@{}} 0.140 \\ (0.003)\end{tabular} & \begin{tabular}[c]{@{}c@{}} 1.527 \\ (0.030)\end{tabular} & \begin{tabular}[c]{@{}c@{}} 15.017 \\ (0.083)\end{tabular} & \begin{tabular}[c]{@{}c@{}} 0.011 \\ (0.000)\end{tabular} \\
\hline
FC2  & \begin{tabular}[c]{@{}c@{}} 0.198 \\ (0.027)\end{tabular} & \begin{tabular}[c]{@{}c@{}} 1.975 \\ (0.077)\end{tabular} & \begin{tabular}[c]{@{}c@{}} 1.158 \\ (0.025)\end{tabular} & \begin{tabular}[c]{@{}c@{}} 11.542 \\ (0.062)\end{tabular} & \begin{tabular}[c]{@{}c@{}} 0.386 \\ (0.009)\end{tabular} \\ \hline
Total  & \begin{tabular}[c]{@{}c@{}} 6.298 \\ (0.291)\end{tabular} & \begin{tabular}[c]{@{}c@{}} 63.706 \\ (1.113)\end{tabular} & \begin{tabular}[c]{@{}c@{}} 154.494 \\ (1.110)\end{tabular} & \begin{tabular}[c]{@{}c@{}} 1537.227 \\ (2.555)\end{tabular} & \begin{tabular}[c]{@{}c@{}} 16.247 \\ (0.103)\end{tabular} \\

\hline
\end{tabular}}
\label{table:inference_experiment1}
\end{center}
\end{table}

We conducted additional experiments to determine the number of samples at which \sysname starts to outperform TenSEAL. The experiment was carried out by incrementally increasing the number of samples and observing the inference time. The results are presented in Figure~\ref{figure:ablation2}.

\begin{figure}[!ht]
\centerline{\includegraphics[clip=true, width=.8\columnwidth]{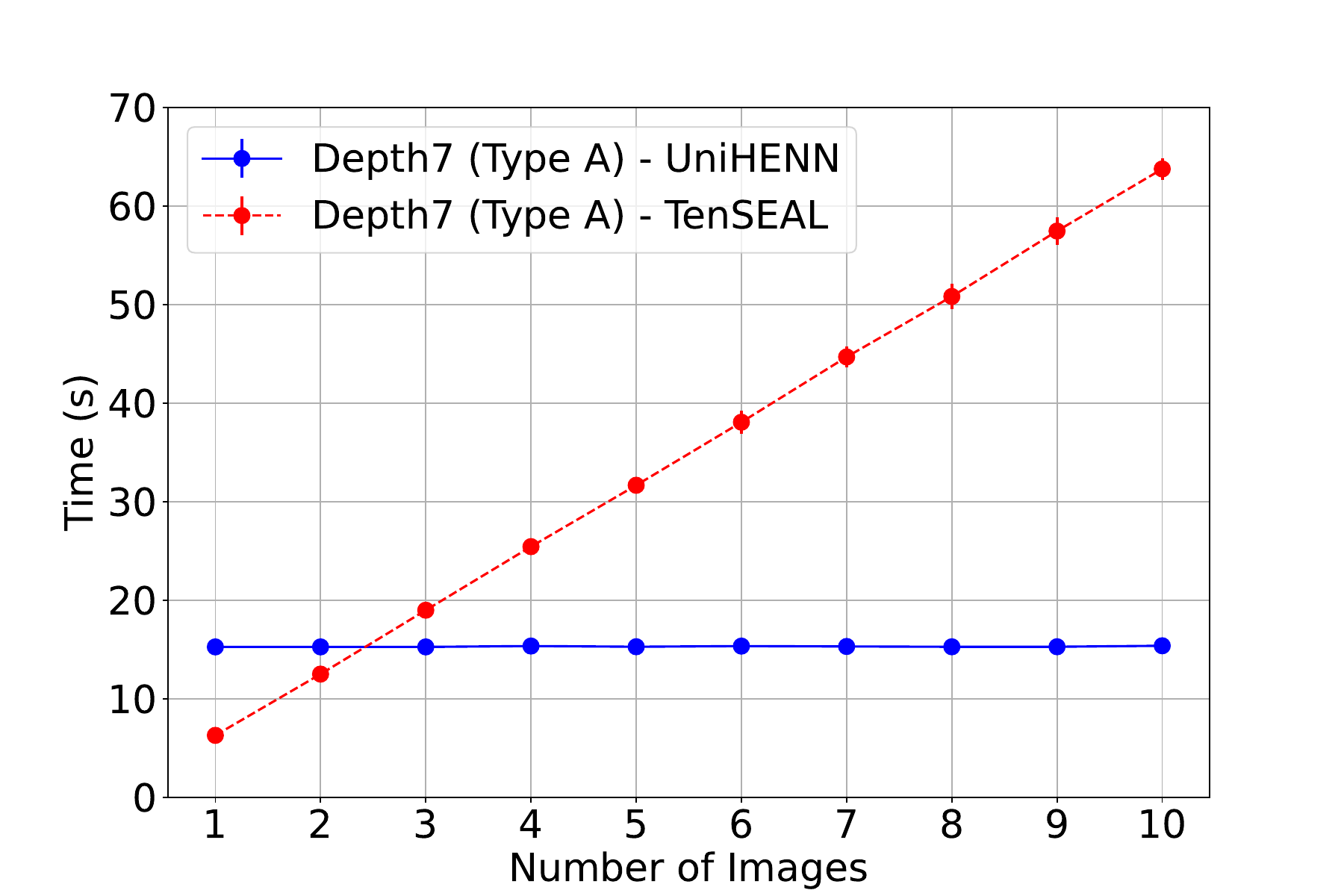}}
\caption{Average inference time (in seconds) for the $M_1$ model architecture using \sysname and TenSEAL with varying the number of input samples.}
\label{figure:ablation2}
\end{figure}

In Figure~\ref{figure:ablation2}, the inference times for \sysname and TenSEAL using the $M_1$ model are presented. Since TenSEAL does not support batched inference, its inference time increases linearly as the number of input images grows. In contrast, \sysname does support batched inference, allowing up to 10 MNIST images to be processed simultaneously in the $M_1$ model. The figure reveals that \sysname's inference time becomes shorter than TenSEAL's starting at a batch size of 3. This demonstrates that \sysname surpasses TenSEAL in efficiency when concurrently processing $k$ images, where $k \geq 3$.

\subsection{Experiment 2: Comparison of Inference Time between \sysname and \pycrcnn for \lenet-1}
\label{experiment2} 

In this experiment, we utilize the LeNet-1 model, denoted as $M_2$. The accuracy of $M_2$ is 98.62\%, slightly higher than $M_1$. The detailed specifications of these model architectures can be found in Table~\ref{table:parameter_setting_M2}. Implementing the hyperbolic tangent function (tanh) in HE is computationally challenging; therefore, we have substituted the activation function from tanh to the square function. We find that the modified LeNet-1 model with the square activation achieves an accuracy of 98.62\% comparable to the 98.41\% accuracy of the original LeNet-1 model when tested on a 10,000 sample MNIST dataset. This indicates that the modification in the activation function does not significantly impact the model's accuracy. 

\begin{table}[!ht]
\begin{center}
\caption{Detailed parameters for $M_2$.}
\renewcommand{\arraystretch}{1.0}
\begin{tabular}{|c|c|c|}
\hline
Layer & Parameter & \#Mult\\
\hline\hline
Conv2d & $CH_{in}$ = 1, $CH_{out}$ = 4, $K$ = 5, $S$ = 1 & 1\\
\hline
Square  & - & 1 \\
\hline
AvgPool2d & kernel size = 2 & 0 \\
\hline
Conv2d & $CH_{in}$ = 4, $CH_{out}$ = 12, $K$ = 5, $S$ = 1 & 1\\
\hline
Square  & - & 1 \\
\hline
AvgPool2d & kernel size = 2 & 0 \\
\hline
Flatten &  - & 2\\
\hline
FC1 & $DAT_{in}$ = 192, $DAT_{out}$ = 10 & 1\\
\hline
Total \#Mult & - & 7 \\
\hline
\end{tabular}
\label{table:parameter_setting_M2}
\end{center}
\end{table}

This experiment aims to confirm that \sysname can support CNN models with more than two convolutional layers. Importantly, the TenSEAL library cannot be used in this experiment due to its constraint of supporting only a single convolutional layer, a limitation arising from its \texttt{im2col} algorithm implementation. The results are presented in Table~\ref{table:inference_experiment2}.

\begin{table}[!ht]
\begin{center}
\caption{Comparison of the average (with standard deviation) inference time in seconds on the MNIST dataset between PyCrCNN and UniHENN for $M_2$ model architecture.}
\renewcommand{\arraystretch}{1.0}
\begin{tabular}{|c|c|c|}
\hline
Model & PyCrCNN & \sysname \\
\hline\hline
Drop Level & - & 0.066 (0.003) \\
\hline
Conv2d & 169.182 (0.891) & 3.715 (0.048) \\
\hline
Square & 88.759 (0.605) & 0.140 (0.003) \\
\hline
AvgPool2d & 1.955 (0.080) & 0.535 (0.012) \\
\hline
Conv2d & 520.452 (1.609) & 21.697 (0.150) \\
\hline
Square & 13.297 (0.116) & 0.260 (0.005) \\
\hline
AvgPool2d & 0.419 (0.019) & 0.875 (0.011) \\
\hline 
Flatten & 0.000 (0.000) & 2.279 (0.020) \\
\hline 
FC1 & 6.527 (0.070) & 0.522 (0.014) \\
\hline 
Total & 800.591 (2.279) & 30.089 (0.155) \\
\hline 
\end{tabular}
\label{table:inference_experiment2}
\end{center}
\end{table}

In Table~\ref{table:inference_experiment2}, \sysname takes 30.089 seconds, making it approximately 26.6 times faster than PyCrCNN, which takes 800.591 seconds. In this experiment, most of the computational time for both \sysname and PyCrCNN is consumed in the second convolutional layer. These findings highlight the importance of optimizing convolutional layer operations for time-efficient CNN inference in the context of HE.

Furthermore, both Experiments 1 and 2 employ the same input ciphertext, showcasing that \sysname enables diverse CNN models without requiring re-encryption, provided the supported HE parameters across the models are identical.


\subsection{Experiment 3: Adaptibility of \sysname for CNN Models with Approximate \relu activation}
\label{experimet3}

In this experiment, we implement a CNN model, denoted as $M_3$, with approximate ReLU for the MNIST dataset. Originally, we planned to modify $M_2$ by replacing the square activation function with approximate ReLU, but the performance was not satisfactory. Therefore, we redesigned the model to achieve better performance.

The model accuracy of $M_3$ is 98.22\%, which is similar to $M_2$'s accuracy of 98.62\%. Table~\ref{table:parameter_setting_M3} provides detailed specifications for this model architecture.

\begin{table}[!ht]
\begin{center}
\caption{Detailed parameters for $M_3$.}
\renewcommand{\arraystretch}{1.0}
\begin{tabular}{|c|c|c|}
\hline
Layer & Parameter & \#Mult\\
\hline\hline
Conv2d & $CH_{in}$ = 1, $CH_{out}$ = 6, $K$ = 3, $S$ = 1 & 1\\
\hline
Approx ReLU  &  $f(x) = 0.375373 + 0.5x + 0.117071x^{2}$ & 2 \\
\hline
AvgPool2d & kernel size = 2 & 0 \\
\hline
Flatten &  - & 2\\
\hline
FC1  & $DAT_{in}$ = 1014, $DAT_{out}$ = 120 & 1 \\
\hline
Approx ReLU  &  $f(x) = 0.375373 + 0.5x + 0.117071x^{2}$ & 2 \\
\hline
FC2  & $DAT_{in}$ = 120, $DAT_{out}$ = 10 & 1 \\
\hline
Total \#Mult & - & 9 \\
\hline
\end{tabular}
\label{table:parameter_setting_M3}
\end{center}
\end{table}

Table~\ref{table:inference_experiment3} shows a layer-by-layer breakdown of the inference time for the $M_3$ model architecture. The Flatten and FC1 layers are the most time-consuming, taking an average of 9.603 seconds and 15.557 seconds, respectively. These two layers alone contribute significantly to the total inference time of 29.105 seconds. While \sysname is significantly slower than non-HE models, it provides enhanced privacy and security, which may be crucial for specific applications or compliance requirements. The inference time of 29.105 seconds is still under 30 seconds, suggesting that \sysname is practical for real-world applications, especially in contexts where data security is paramount. This inference time could be considered acceptable depending on the specific use case and the sensitivity of the data being processed. These results further validate the adaptability and efficiency of \sysname in handling various CNN architectures with customized functionalities like approximate ReLU.

\begin{table}[!ht]
\begin{center}
\caption{Average (with standard deviation) inference time in seconds on the MNIST dataset for the $M_3$ model architecture.}
\renewcommand{\arraystretch}{1.0}
\begin{tabular}{|c|c|}
\hline
Model &  UniHENN \\
\hline\hline
Drop Level & 0.035 (0.001) \\
\hline
Conv2d & 1.863 (0.029) \\
\hline
Approx ReLU & 0.513 (0.011) \\
\hline
AvgPool2d & 1.030 (0.024) \\
\hline
Flatten & 9.603 (0.081) \\
\hline 
FC1 & 15.557 (0.111) \\
\hline 
Approx ReLU & 0.045 (0.002) \\
\hline
FC2 & 0.459 (0.010) \\
\hline 
Total & 29.105 (0.171) \\
\hline 
\end{tabular}
\label{table:inference_experiment3}
\end{center}
\end{table}

\subsection{Experiment 4: Adaptibility of \sysname for \lenet-5}
\label{experiment4}

In this experiment, we implement the LeNet-5 model~\cite{lecun1998gradient}, denoted as $M_4$, with the only modification being the activation function, which is changed from tanh to square. The model accuracy of $M_4$ is 98.91\%. Table~\ref{table:parameter_setting_M4} provides detailed specifications for this model architecture.

\begin{table}[!ht]
\begin{center}
\caption{Detailed parameters for $M_4$.}
\renewcommand{\arraystretch}{1.0}
\begin{tabular}{|c|c|c|}
\hline
Layer & Parameter & \#Mult\\
\hline\hline
Conv2d & $CH_{in}$ = 1, $CH_{out}$ = 6, $K$ = 5, $S$ = 1 & 1\\
\hline
Square  & - & 1 \\
\hline
AvgPool2d & kernel size = 2 & 0 \\
\hline
Conv2d & $CH_{in}$ = 6, $CH_{out}$ = 16, $K$ = 5, $S$ = 1 & 1\\
\hline
Square  & - & 1 \\
\hline
AvgPool2d & kernel size = 2 & 0 \\
\hline
Conv2d & $CH_{in}$ = 16, $CH_{out}$ = 120, $K$ = 5, $S$ = 1 & 1\\
\hline
Square  & - & 1 \\
\hline
Flatten &  - & 0\\
\hline
FC1 & $DAT_{in}$ = 120, $DAT_{out}$ = 84 & 1\\
\hline
Square  & - & 1 \\
\hline
FC2  & $DAT_{in}$ = 84, $DAT_{out}$ = 10 & 1 \\
\hline
Total \#Mult & - & 9 \\
\hline
\end{tabular}
\label{table:parameter_setting_M4}
\end{center}
\end{table}

Table~\ref{table:inference_experiment4} shows a layer-by-layer breakdown of the inference time for the $M_4$ model architecture. The results show that the inference time for $M_4$ is 740.128 seconds. This is substantially slower than the previous experiments, which could be attributed to the model's increased complexity. The slow performance emphasizes the need for further optimization, especially if \sysname is to be broadly applied to more complex CNN architectures like LeNet-5 for real-world applications. Note that the convolutional layers are the most time-consuming, highlighting the critical need for optimizing these operations when implementing CNNs using HE. With such a long inference time, the immediate practicality of using this model for real-time or near-real-time applications is limited. However, this could be acceptable for services that require strong privacy controls and where data security is a higher priority than speed. For example, in healthcare or financial services, where data may be extremely sensitive, this level of privacy may justify the slower inference times.

\begin{table}[!ht]
\begin{center}
\caption{Average (with standard deviation) inference time in seconds on the MNIST dataset for the $M_4$ model architecture.}
\renewcommand{\arraystretch}{1.0}
\begin{tabular}{|c|c|}
\hline
Layer & \sysname \\
\hline\hline
Drop Level & 0.035 (0.002) \\
\hline
Conv2d & 5.244 (0.048) \\ 
\hline
Square  & 0.312 (0.005) \\ 
\hline
AvgPool2d & 0.865 (0.017) \\ 
\hline
Conv2d & 48.222 (0.266) \\ 
\hline
Square  & 0.567 (0.010) \\ 
\hline
AvgPool2d & 1.453 (0.020) \\ 
\hline
Conv2d & 668.688 (1.241) \\ 
\hline
Square  & 2.613 (0.030) \\ 
\hline
Flatten & 3.995 (0.044) \\ 
\hline
FC1 & 7.669 (0.074) \\
\hline
Square  & 0.011 (0.000) \\ 
\hline
FC2 & 0.454 (0.009) \\ 
\hline
Total & 740.128 (1.381) \\ 
\hline
\end{tabular}
\label{table:inference_experiment4}
\end{center}
\end{table}

\subsection{Experiment 5: Adaptability of \sysname for CNN Models on Color Images}
\label{experiment5} 

In this experiment, we evaluate the adaptability of \sysname using the CIFAR-10 color image dataset. We implement a CNN model, denoted as \( M_5 \), which is a modified version of $M_4$, to achieve satisfactory accuracy on CIFAR-10. The model accuracy of $M_5$ is 73.26\%. Table~\ref{table:parameter_setting_M5} provides detailed specifications for this model architecture.

Note that a comparison with TenSEAL is not feasible, as TenSEAL does not process multiple channels, which is essential to process color images. Additionally, we attempted an experiment with PyCrCNN under the same settings but failed to obtain results despite running the experiment for approximately 15 hours. This failure is attributed to PyCrCNN's approach of encrypting each parameter with an individual ciphertext, which demands substantial memory and computational time.

\begin{table}[!ht]
\begin{center}
\caption{Detailed parameters for $M_5$.}
\resizebox{0.99\linewidth}{!}{
\renewcommand{\arraystretch}{1.0}
\begin{tabular}{|c|c|c|}
\hline
Layer & Parameter & \#Mult\\
\hline\hline
Conv2d & $CH_{in}$ = 3, $CH_{out}$ = 16, $K$ = 3, $S$ = 1 & 1\\
\hline
Square  & - & 1 \\
\hline
AvgPool2d & kernel size = 2 & 0 \\
\hline
Conv2d & $CH_{in}$ = 16, $CH_{out}$ = 64, $K$ = 4, $S$ = 1 & 1\\
\hline
Square  & - & 1 \\
\hline
AvgPool2d & kernel size = 2 & 0 \\
\hline
Conv2d & $CH_{in}$ = 64, $CH_{out}$ = 128, $K$ = 3, $S$ = 1 & 1\\
\hline
Square  & - & 1 \\
\hline
AvgPool2d & kernel size = 4 & 0 \\
\hline
Flatten &  - & 1\\
\hline
FC1 & $DAT_{in}$ = 128, $DAT_{out}$ = 10 & 1\\
\hline
Total \#Mult & - & 8 \\
\hline
\end{tabular}}
\label{table:parameter_setting_M5}
\end{center}
\end{table}

Table~\ref{table:inference_experiment5} presents the experimental results, showing that the total inference time for \sysname is approximately 21 minutes. While this may seem long, it is important to note that the convolutional layers are particularly time-consuming, requiring about 256.000 seconds for the first layer and 938.446 seconds for the second. This is consistent with the results of previous experiments. Despite the long inference time, we consider it tolerable given the inherent complexities of performing inference on color images, a capability not offered by alternative solutions.

\begin{table}[!ht]
\begin{center}
\caption{Average (with standard deviation) inference time in seconds for the $M_5$ model architecture.}
\renewcommand{\arraystretch}{1.0}
\begin{tabular}{|c|c|}
\hline
Layer & \sysname \\
\hline\hline
Drop Level & 0.155 (0.004) \\
\hline
Conv2d & 8.935 (0.062) \\ 
\hline
Square  & 0.703 (0.011) \\ 
\hline
AvgPool2d & 0.429 (0.004) \\ 
\hline
Conv2d & 256.000 (0.867) \\ 
\hline
Square  & 1.839 (0.023) \\ 
\hline
AvgPool2d & 1.891 (0.022) \\ 
\hline
Conv2d & 938.446 (1.884) \\ 
\hline
Square  & 2.064 (0.019) \\ 
\hline
AvgPool2d  & 4.466 (0.049) \\ 
\hline
Flatten & 33.370 (0.162) \\ 
\hline
FC1 & 3.146 (0.025) \\
\hline
Total & 1251.444 (2.532) \\ 
\hline
\end{tabular}
\label{table:inference_experiment5}
\end{center}
\end{table}

\subsection{Experiment 6: Evaluation of Inference Time of \sysname on Grayscale Images}
\label{experiment6}

To assess the adaptability of \sysname to diverse datasets, we performed experiments using the USPS dataset and the $M_6$ model architecture. For performance comparison, we also implemented the same model using PyCrCNN. The model accuracy of $M_6$ is 94.21\%. Table~\ref{table:parameter_setting_M6} provides detailed specifications for this model architecture. 

\begin{table}[!ht]
\begin{center}
\caption{Detailed parameters for $M_6$.}
\renewcommand{\arraystretch}{1.0}
\begin{tabular}{|c|c|c|}
\hline
Layer & Parameter & \#Mult\\
\hline\hline
Conv2d & $CH_{in}$ = 1, $CH_{out}$ = 6, $K$ = 4, $S$ = 2 & 1\\
\hline
Square  & - & 1 \\
\hline
Flatten &  - & 2\\
\hline
FC1  & $DAT_{in}$ = 294, $DAT_{out}$ = 64 & 1 \\
\hline
Square  & - & 1 \\
\hline
FC2  & $DAT_{in}$ = 64, $DAT_{out}$ = 10 & 1 \\
\hline
Total \#Mult & - & 7 \\
\hline
\end{tabular}
\label{table:parameter_setting_M6}
\end{center}
\end{table}

Table~\ref{table:inference_experiment6} shows that \sysname is approximately 6 times faster than PyCrCNN, with a total inference time of 12.309 seconds compared to PyCrCNN's 71.139 seconds. Interestingly, the FC1 layer in PyCrCNN consumes a significant portion of the time (42.792 seconds), likely due to its large input size of 294. In contrast, although \sysname also incurs a relatively high computational cost at the FC1 layer, it is substantially less than that of PyCrCNN. This suggests that \sysname can efficiently handle FC layers with large input sizes.

\begin{table}[!ht]
\begin{center}
\caption{Comparison of the average (with standard deviation) inference time in seconds on the USPS dataset between PyCrCNN and \sysname for the $M_6$ model architecture.}
\renewcommand{\arraystretch}{1.0}
\begin{tabular}{|c|c|c|}
\hline
Layer & PyCrCNN & \sysname\\
\hline\hline
Drop Level & - & 0.066 (0.002) \\
\hline
Conv2d & 14.441 (0.125) & 2.662 (0.029)\\ 
\hline
Square  & 11.238 (0.090) & 0.212 (0.004)\\ 
\hline
Flatten & 0.000 (0.000) & 3.026 (0.032)\\
\hline
FC1 & 42.792 (0.121) & 5.948 (0.074)\\
\hline
Square & 1.502 (0.026) & 0.011 (0.000)\\ 
\hline
FC2 & 1.166 (0.015) & 0.384 (0.009)\\ 
\hline
Total & 71.139 (0.198) & 12.309 (0.092)\\ 
\hline
\end{tabular}
\label{table:inference_experiment6}
\end{center}
\end{table}

\subsection{Experiment 7: Evaluation of Inference Time of \sysname for a 1D CNN Model}
\label{experiment7} 

To evaluate the adaptability of \sysname for 1D CNN models, we implemented a 1D CNN model, denoted as $M_7$, for processing ECG data. This model, a modified version of the 1D CNN model by Abuadbba et al.~\cite{abuadbba2020use}, achieves an accuracy of 96.87\%, which is comparable to the 98.90\% achieved by Abuadbba et al.'s original model. For performance comparison, the same model was also implemented using PyCrCNN. Table~\ref{table:parameter_setting_M7} provides detailed specifications for this model architecture. 


\begin{table}[!ht]
\begin{center}
\caption{Detailed parameters for $M_7$.}
\renewcommand{\arraystretch}{1.0}
\begin{tabular}{|c|c|c|}
\hline
Layer & Parameter & \#Mult\\
\hline\hline
Conv1d & $CH_{in}$ = 1, $CH_{out}$ = 2, $K$ = 2, $S$ = 2 & 1\\
\hline
Square & - & 1\\
\hline
Conv1d & $CH_{in}$ = 2, $CH_{out}$ = 4, $K$ = 2, $S$ = 2 & 1\\
\hline
Flatten & - & 1\\
\hline
FC1  & $DAT_{in}$ = 128, $DAT_{out}$ = 32 & 1 \\
\hline
Square  & - &  1\\
\hline
FC2  & $DAT_{in}$ = 32, $DAT_{out}$ = 5 & 1 \\
\hline
Total \#Mult & - & 7\\
\hline
\end{tabular}
\label{table:parameter_setting_M7}
\end{center}
\end{table}

Table~\ref{table:inference_experiment7} shows that \sysname is approximately 3.0 times faster than PyCrCNN, recording a total inference time of 5.119 seconds compared to PyCrCNN's 15.514 seconds. This indicates that \sysname is also more efficient even for relatively smaller models, such as 1D CNNs, where the computational time for convolutional layers is not as extensive. This efficiency positions \sysname as an ideal choice for privacy-sensitive applications like disease diagnosis systems, aligning well with the principles of HE.

\begin{table}[!ht]
\begin{center}
\caption{Average (with standard deviation) inference time in seconds for the $M_7$ model architecture on the ECG dataset.}
\renewcommand{\arraystretch}{1.0}
\begin{tabular}{|c|c|c|}
\hline
Layer & PyCrCNN & \sysname\\
\hline\hline
Drop Level & - & 0.069 (0.007) \\
\hline
Conv1d & 0.556 (0.011) & 0.113 (0.003)\\ 
\hline
Square  & 4.913 (0.043) & 0.075 (0.002)\\ 
\hline
Conv1d  & 1.081 (0.041) & 0.283 (0.006)\\ 
\hline
Flatten & 0.000 (0.000) & 1.622 (0.015)\\ 
\hline
FC1  & 8.163 (0.067) & 2.765 (0.022)\\ 
\hline
Square & 0.559 (0.011) & 0.012 (0.000)\\
\hline
FC2 & 0.242 (0.006) & 0.180 (0.016)\\ 
\hline
Total & 15.514 (0.114) & 5.119 (0.039)\\ 
\hline
\end{tabular}
\label{table:inference_experiment7}
\end{center}
\end{table}

\subsection{Validation of inference results on encrypted data}

The CKKS scheme operates on approximate complex arithmetic, which can introduce minor errors after homomorphic operations. Therefore, it is crucial to validate \sysname's inference results by comparing them with the outcomes of the original plaintext inference.

We selected $M_4$ and $M_5$ as representative models for validation, as the errors were very small in the case of other models. We evaluated the $M_4$ model with 2,000 samples from the MNIST dataset and the $M_5$ model with 2,000 samples from the CIFAR-10 dataset. The results showed that all outputs produced by \sysname were equivalent to those of the original models without any significant loss in accuracy. This indicates that \sysname can perform secure and highly accurate inferences under the parameter configurations presented in Table~\ref{table:parameter_setting}.


\subsection{Impact of HE Parameters}
\label{Effects of the Parameters for HE}

We conducted an additional analysis to investigate how the inference time and encrypted inference error vary depending on the parameters of the CKKS scheme.

The CKKS scheme's operation time and decrypted result accuracy are influenced by several parameters, specifically \textit{\# slots}, \textit{scale}, and \textit{depth}. Therefore, the inference time for CNN models is intricately linked to these parameters. A detailed analysis of these parameter settings is essential to optimize inference time while maintaining result accuracy. The \textit{\# slots} parameter is determined by the degree $N$ of the polynomial ring, where $N$ is half the ring's value. The \textit{scale} parameter represents the precision of floating-point digits, while \textit{depth} is influenced by both \textit{scale} and the polynomial $P$.

\subsubsection{Influence of Depth on Inference Time}
\label{experiment9}

The operation time increases as the \textit{depth} increases while using a fixed \textit{scale}. Figure~\ref{figure:ablation3} shows that, in the $M_1$ model, the inference time increases linearly with \textit{depth}. This is attributed to the ciphertext size being dependent on \textit{depth}, thereby increasing the computational workload. To achieve optimal performance, the \textit{depth} should align with the number of required multiplications for the model. However, in our architecture, \textit{depth} is predetermined before model selection. To address this, we set the \textit{depth} as high as possible and then optimize by fine-tuning the input ciphertext.

\begin{figure}[!ht]
\centerline{\includegraphics[clip=true, width=.8\columnwidth]{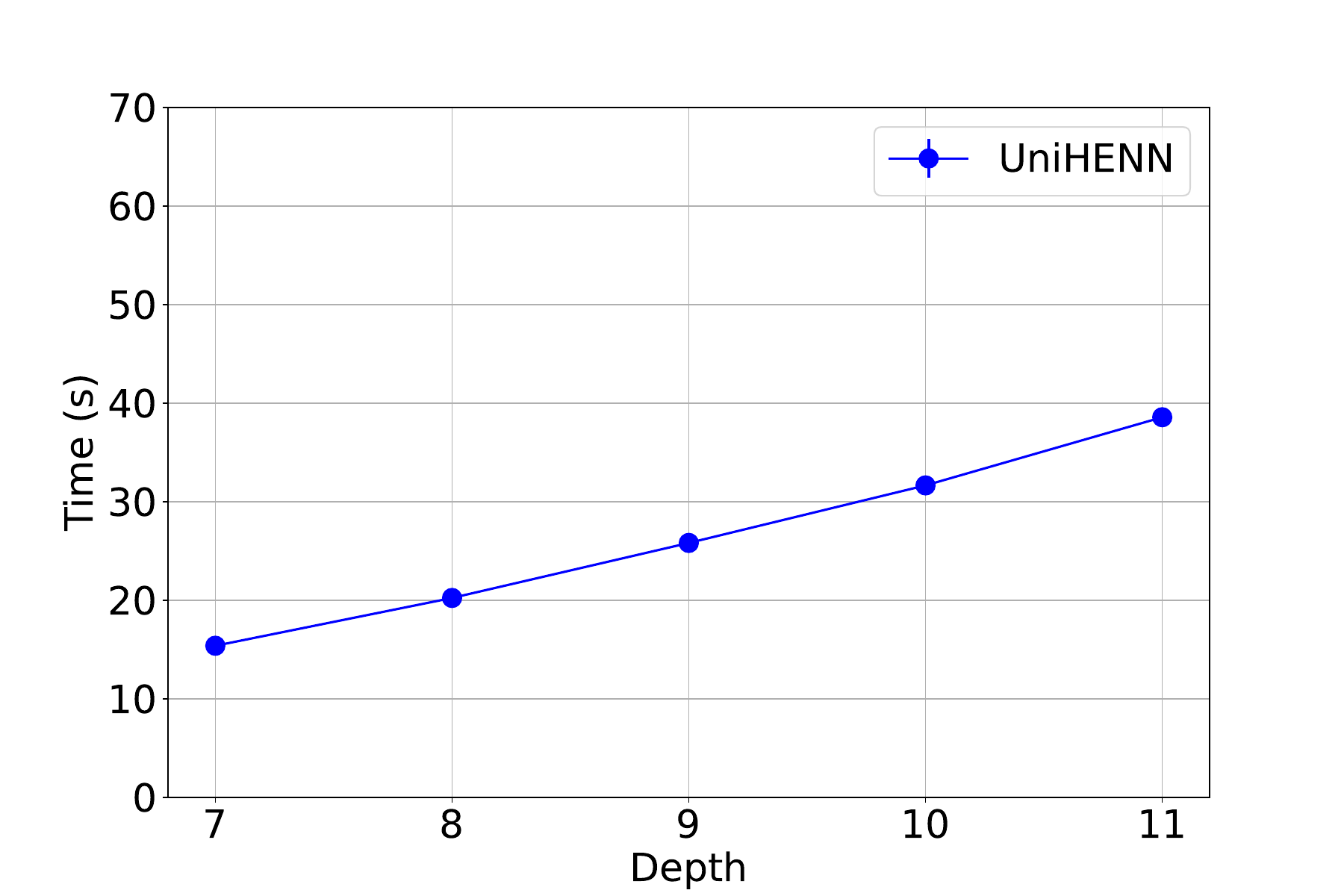}}
\caption{Average inference time for $M_1$ with \textit{depth}.}
\label{figure:ablation3}
\end{figure}


\subsubsection{Impact of Scale on Error}
\label{Impact of Scale on Error}

We measured the error in the \(M_1\) model's results as a function of \textit{scale}. Detailed findings are presented in Figure~\ref{figure:ablation4}. Increasing \textit{scale} logarithmically reduces the error. Although a higher \textit{scale} is advantageous, it consequently leads to a lower \textit{depth}, given their inverse relationship due to the fixed \textit{log Q} parameter in the CKKS scheme. Therefore, selecting an optimal \textit{scale} is crucial for minimizing error while ensuring sufficient \textit{depth} for CNN inference. A \textit{scale} of 32 ensures 32-bit decimal point precision. From our observations, we note that the error converges toward zero as the \textit{scale} increases. Specifically, using a \textit{scale} value of 30 or higher can significantly reduce the error.

\begin{figure}[!ht]
\centerline{\includegraphics[clip=true, width=.8\columnwidth]{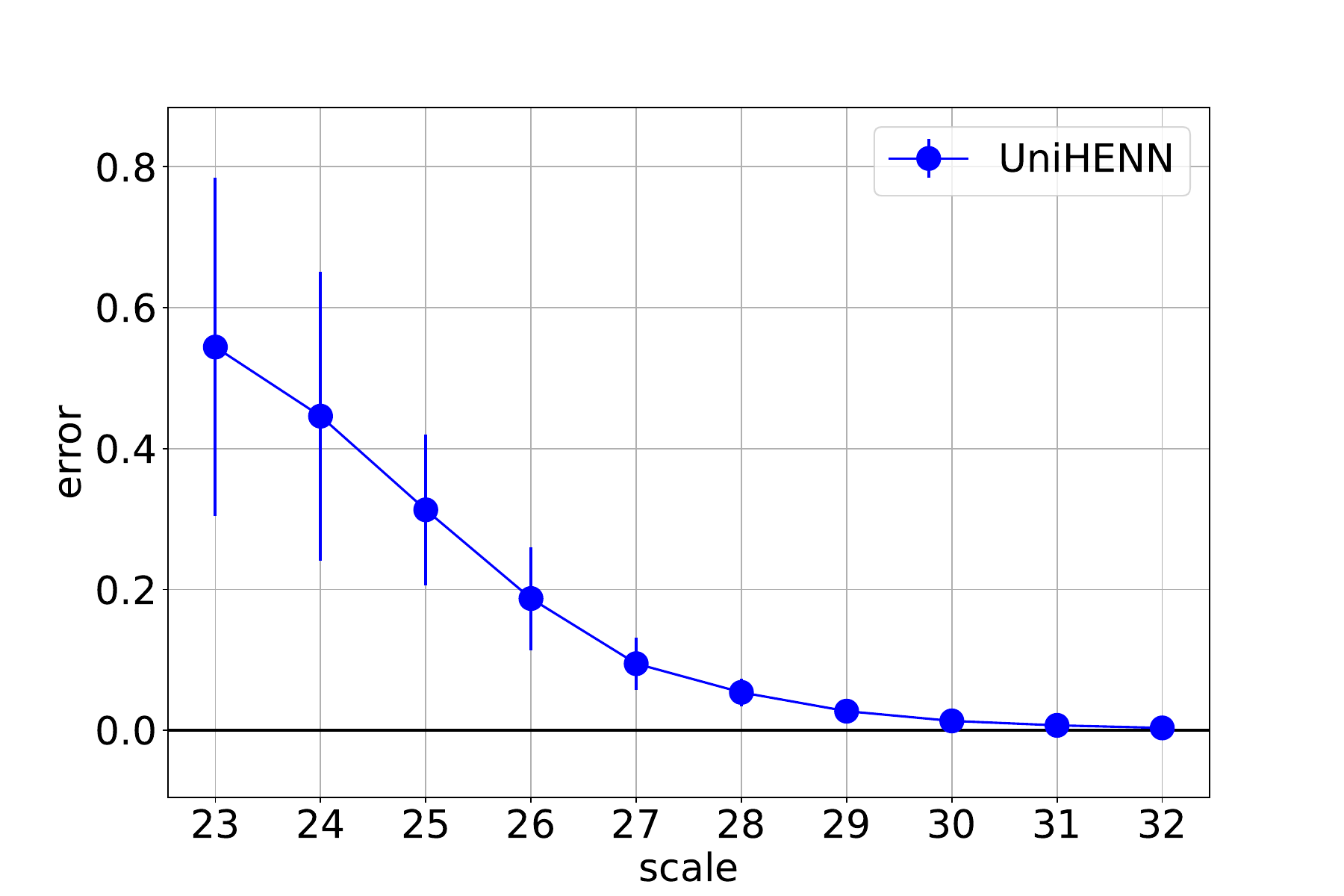}}
\caption{Error in $M_1$ with \textit{scale}.}
\label{figure:ablation4}
\end{figure}

\subsection{Comparison of \sysname and DP-based CNN}
\label{sec:comparison_unihenn_opacus}

This section compares our \sysname framework with differential privacy (DP) for privacy-preserving CNNs, providing insights into the trade-offs between HE and DP-based approaches. We implemented DP-SGD~\cite{abadi2016deep} using Opacus~\cite{opacus}.

We conducted experiments on the MNIST dataset with model $M_{4}$, comprising 60,000 training and 10,000 test images. The Adam optimizer was used with a learning rate of 0.001, training for 15 epochs with a batch size of 32. For the DP-SGD implementation, we set the maximum per gradient norm to 1.0 and used $\delta=0.00001$, satisfying the condition $\delta < 1/60,000$ for the MNIST training set size. We varied $\epsilon$ from 5 to 0.1 during the 15-epoch training. While there is no definitive guideline for choosing $\epsilon$, NIST's post~\cite{nistdppost} suggests that $\epsilon$ between 0 and 5 provides strong privacy protection. Thus, our selected $\epsilon$ values with $\delta=0.00001$ would represent suitable candidates for DP-SGD-based training.

Tables \ref{table:accuracy_between_dp_and_unihenn2} and \ref{table:memory_between_dp_and_unihenn2} present the accuracy and memory consumption results, respectively.

\begin{table}[h]
\centering
\caption{Accuracy comparison between \sysname (U) and DP-based CNN (D) for various privacy budgets ($\epsilon$)  (30 runs, standard deviations in parentheses).}
\resizebox{0.99\linewidth}{!}{
\renewcommand{\arraystretch}{1.0}
\begin{tabular}{|c|c|c|c|c|c|}
\hline
Framework & U & D ($\epsilon$=5) & D ($\epsilon$=1) & D ($\epsilon$=0.5) & D ($\epsilon$=0.1)\\
\hline
Accuracy (\%) & 98.91 & 84.19 (2.59) & 84.57 (2.29) & 83.73 (1.98) &80.63 (1.54) \\
\hline
\end{tabular}}
\label{table:accuracy_between_dp_and_unihenn2}
\end{table}

\begin{table}[h]
\centering
\caption{Memory consumption comparison during inference (30 runs, standard deviations in parentheses).}
\renewcommand{\arraystretch}{1.0}
\begin{tabular}{|c|c|c|}
\hline
Framework & U & D \\
\hline
Memory (MiB) & 3677.64 (1.05) & 3.91 (0.10) \\
\hline
\end{tabular}
\label{table:memory_between_dp_and_unihenn2}
\end{table}

Our results reveal significant trade-offs between HE-based and DP-based privacy-preserving techniques in CNNs. \sysname achieves superior accuracy (98.91\%) compared to the DP-based CNN, which ranges from 84.57\% to 80.63\% as $\epsilon$ decreases. This demonstrates \sysname's ability to maintain model performance across varying privacy parameters. However, this high accuracy comes at a substantial computational cost: \sysname consumes about 941 times more memory (3677.64 MiB) than the DP-based approach (3.91 MiB) due to the intensive nature of HE operations. 

These findings highlight the critical resource trade-offs in privacy-preserving machine learning, emphasizing the need to balance accuracy and computational efficiency when selecting privacy-preserving techniques for specific applications. The choice between HE and DP approaches will depend on the particular use case, with \sysname offering superior accuracy at the cost of higher computational resources, while DP-based methods provide a more memory-efficient solution with a trade-off in accuracy.

\section{Limitations}
\label{limitations}

We propose an optimized CNN model inference mechanism based on HE, \sysname. While \sysname uses batching to reduce inference time, it has three key limitations.

Firstly, complex deep learning models often require a large number of operations, particularly multiplications. In HE, multiplication operations must be limited based on parameter values, or bootstrapping must be used. Due to the extensive multiplication operations in deep learning models, bootstrapping is necessary, significantly slowing computation speed.

Secondly, \sysname does not currently support multi-core architectures or GPUs. Our deployment uses SEAL-Python, an open-source HE framework suited for CPU and single-thread architectures. Most open-source HE frameworks, including SEAL-Python, are designed this way. Recent studies have optimized HE operations for GPUs~\cite{al2020multi, ozcan2023homomorphic}, focusing on parallelism to enhance performance. To address these limitations, we plan to extend our architecture to support multi-core and GPU environments. Our future work will involve optimizing key HE operations for parallel execution on GPUs, implementing efficient memory management techniques, and ensuring proper synchronization mechanisms to maximize computational performance.

Finally, \sysname is designed to reduce computational time through batching, making it efficient for analyzing sensitive data at scale or when inferring multiple data points simultaneously. However, due to its focus on batch operations, \sysname may be inefficient for inferring single data points. To overcome this, we plan to improve computational efficiency by distributing operations within the convolutional and FC layers in the space where batch operations are performed, thereby increasing overall efficiency.
\section{Related Work}
\label{sec:related work}

\subsection{Libraries of HE}
Several prominent HE libraries are available, including SEAL~\cite{sealcrypto}, Lattigo~\cite{lattigo}, HElib~\cite{helib}, and OpenFHE~\cite{openFHE}. Each offers unique features: SEAL~\cite{sealcrypto} supports BFV and CKKS schemes with optimizations like ciphertext packing. We used SEAL-Python, a Python port of the C++ implementation, for our work. Lattigo~\cite{lattigo} implements Ring-LWE-based HE primitives and Multiparty-HE protocols in Go, enabling cross-platform builds and optimized secure computation. HElib~\cite{helib} supports BGV and CKKS schemes with bootstrapping, focusing on Smart-Vercauteren ciphertext packing and Gentry-Halevi-Smart optimization. It also provides an HE assembly language for low-level control. OpenFHE~\cite{openFHE} offers bootstrapping and hardware acceleration using a standard Hardware Abstraction Layer (HAL).

\subsection{Machine learning implementation with HE}
To provide privacy-preserving ML services, many previous studies implement HE-friendly neural networks. They focus on two challenges: ``How to build HE-friendly model architecture, especially, the activation layer'' and ``How to build an efficient HE-based CNN (or DNN) system.''

In the first case, there were efforts to replace activation functions with polynomial functions. CryptoNets~\cite{gilad2016cryptonets} use a square function to replace the ReLU activation function. Chou et al.~\cite{chou2018faster} proposed some techniques to replace activation functions such as ReLU, Softplus, and Swish with low-degree polynomials and use pruning and quantization for the efficiency of homomorphic operations. Ishiyama et al.~\cite{HE_Accurate_CNN} adopt Google's Swish activation function with two- and four-degree polynomials with batch normalization to reduce the errors between Swish and approximated polynomials. However, replacing activation functions with low-degree polynomials reduces model accuracy and is only applicable to shallow models. To overcome these limitations, Park et al.~\cite{park2022aespa} proposed the HerPN block, which can replace the batch normalization and ReLU block by utilizing the Hermite polynomial. Recently, Lee et al.~\cite{lee2023optimizations} proposed PP-DNN, a low-latency model optimization solution focusing on convolution and approximate ReLU operations without bootstrapping, supporting CNN models like ResNet-34 with HE. In contrast, \sysname offers comprehensive optimizations across all layer types, including convolutional, pooling, fully connected, and flatten layers, resulting in more efficient operations overall. Additionally, \sysname introduces a batch system for multi-input inference and optimizes element sizes for batch operations, which are not addressed in PP-DNN.

In the second case, many studies~\cite{juvekar2018gazelle, park2022aespa, ghodsi2020cryptonas} interact with clients to compute non-linear operations. The server sends the encrypted intermediate results to the client, and the client computes the non-linear functions with decrypted intermediate results. In this way, the model accuracy can be preserved, but the communication overhead is increased.
\section{Conclusion}
\label{sec:conclusion}

This paper presents \sysname, a privacy-preserving machine learning framework utilizing HE without the \texttt{im2col} operation, enhancing compatibility with diverse ML models. We evaluated \sysname on four public datasets using six 2D CNNs and one 1D CNN, demonstrating accuracy comparable to unencrypted models. \sysname's batch processing technique significantly improves efficiency, outperforming the state-of-the-art TenSEAL library by 3.9 times when processing 10 MNIST images using a simple CNN model. This framework opens up new possibilities for privacy-preserving machine learning in real-world applications. \EOD


\bibliographystyle{plain}
\bibliography{Ref.bib}
       
\begin{IEEEbiography}[{\includegraphics[width=1in,height=1.25in,clip,keepaspectratio]{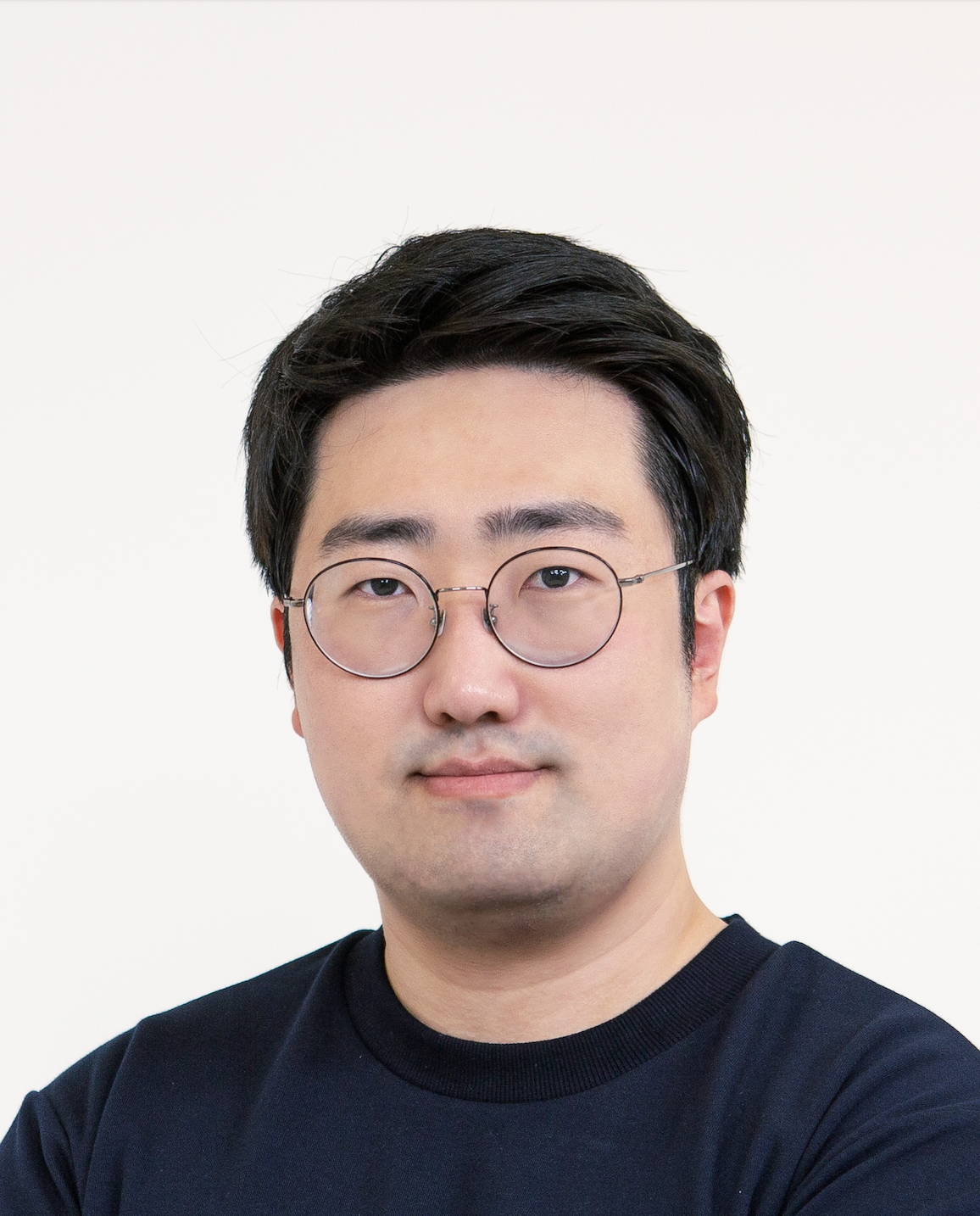}}]{\uppercase{Hyunmin Choi}} 
received the M.S. degree from the Graduate School of Information and Communications, Sungkyunkwan University, Seoul, Republic of Korea, in 2022. He is currently pursuing a Ph.D. degree in the Electrical and Computer Engineering Department at Sungkyunkwan University and working on security development at NAVER Cloud Corporation. His research interests include privacy-enhancing technology and AI security.
\end{IEEEbiography}

\begin{IEEEbiography}[{\includegraphics[width=1in,height=1.25in,clip,keepaspectratio]{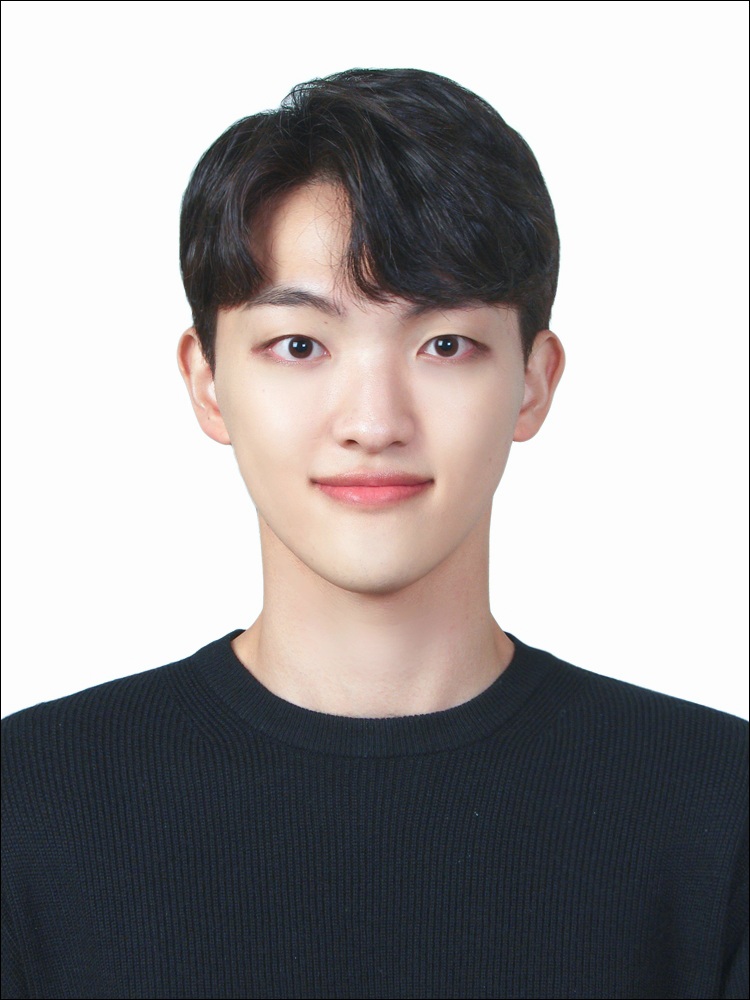}}]{\uppercase{Jihun Kim}} 
is currently pursuing a B.S. degree with dual majors in the Department of Mathematics and the Department of Software at Sungkyunkwan University.
His research interests include homomorphic encryption, data privacy and AI security.
\end{IEEEbiography}

\begin{IEEEbiography}[{\includegraphics[width=1in,height=1.25in,clip,keepaspectratio]{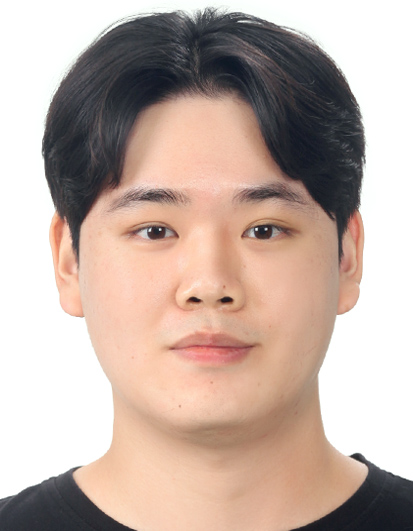}}]{\uppercase{Seungho Kim}}   
is currently pursuing an M.S. degree with the Department of Electrical and Computer Engineering at Sungkyunkwan University. His research interests include data privacy and usable security.
\end{IEEEbiography}

\begin{IEEEbiography}[{\includegraphics[width=1in,height=1.25in,clip,keepaspectratio]{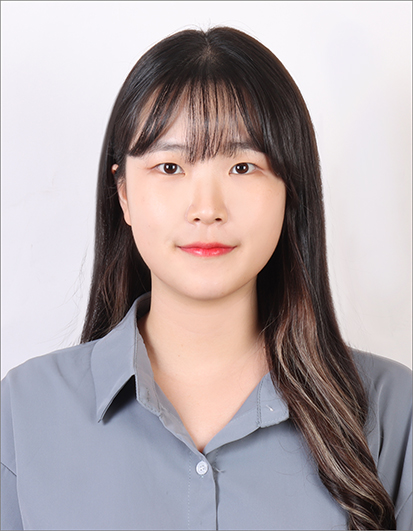}}]{\uppercase{Seonhye Park}} 
is currently pursuing a Ph.D. degree in the Department of Electrical and Computer Engineering at Sungkyunkwan University. Her research interests include AI security and usable security.
\end{IEEEbiography}

\begin{IEEEbiography}[{\includegraphics[width=1in,height=1.25in,clip,keepaspectratio]{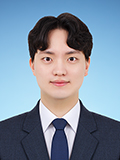}}]{\uppercase{Jeongyong Park}} 
received the M.S. degree of science in engineering from the Department of Software, Sungkyunkwan University, in 2023. His current research interests include AI security and data-driven security.
\end{IEEEbiography}

\begin{IEEEbiography}[{\includegraphics[width=1in,height=1.25in,clip,keepaspectratio]{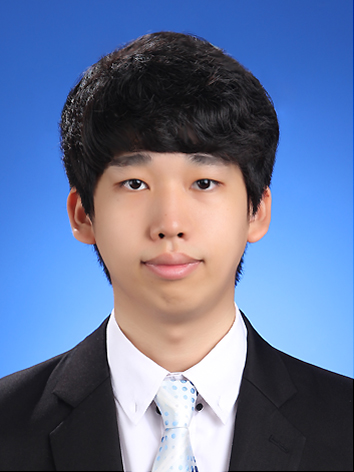}}]{\uppercase{Wonbin Choi}}   
received the M.S. degree in information security from the Department of Information Security, School of Cybersecurity, Korea University, Seoul, South Korea, in 2019. He is currently working in Security Development at NAVER Cloud Corporation. His current research interests include analysis of security vulnerability, formal verification, and information security.
\end{IEEEbiography}

\begin{IEEEbiography}[{\includegraphics[width=1in,height=1.25in,clip,keepaspectratio]{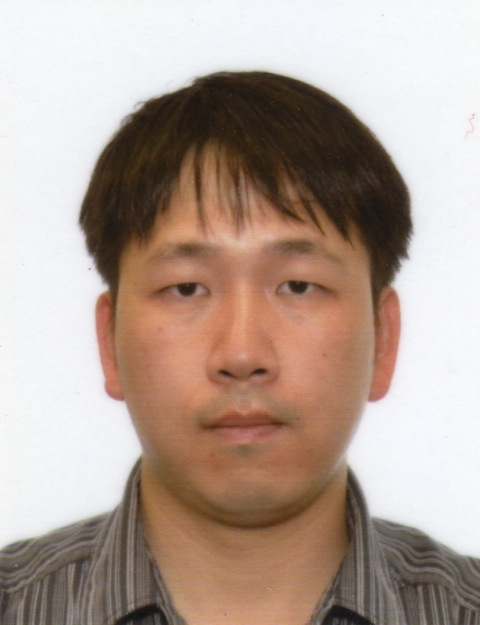}}]{\uppercase{Hyoungshick Kim}} received his B.S. in Information Engineering from Sungkyunkwan University (1999), M.S. in Computer Science from KAIST (2001), and Ph.D. from the University of Cambridge's Computer Laboratory (2012). He completed postdoctoral work at the University of British Columbia's Department of Electrical and Computer Engineering. Kim worked as a senior engineer at Samsung Electronics (2004--2008) and served as a distinguished visiting researcher at CSIRO, Data61 (2019--2020). He is an associate professor in the Department of Computer Science and Engineering at Sungkyunkwan University. His research focuses on usable security, security vulnerability analysis, and data-driven security.
\end{IEEEbiography}


\end{document}